\documentclass[journal]{IEEEtran}
\usepackage{cite}
\usepackage{amsmath,amssymb,amsfonts}
\usepackage{algorithm}
\usepackage{algpseudocode}
\usepackage{graphicx}
\usepackage{xcolor}
\usepackage{svg}
\usepackage[utf8]{inputenc}
\usepackage{textcomp}
\usepackage{stfloats}
\usepackage{multirow}
\usepackage{tabularx}
\usepackage{subcaption}
\usepackage{diagbox}
\usepackage{mathtools}
\usepackage{booktabs}
\usepackage{makecell}
\usepackage{bbm}

\algrenewcommand\algorithmicdo{}
\algrenewtext{EndFor}{\textbf{End}}
\algrenewcommand\algorithmicfor{\textbf{For}}
\usepackage{etoolbox}

\newcolumntype{C}[1]{>{\centering\arraybackslash}p{#1}}

\newtheorem{proposition}{Proposition}

\def\BibTeX{{\rm B\kern-.05em{\sc i\kern-.025em b}\kern-.08em
    T\kern-.1667em\lower.7ex\hbox{E}\kern-.125emX}}

\begin{document}
\title{{Scalable Multimodal Beam Alignment in V2X: An Anti-Imbalance Graph Learning Approach}

\thanks{
Parts of this work have been presented at the 2025 IEEE/CIC International Conference on Communications in China (ICCC) \cite{LWWL-ICCC-25}, Shanghai, China.



}

}

\author{Jiahui Liang,
        Shuoyao Wang,~\IEEEmembership{Senior Member,~IEEE},
        and~Shijian Gao,~\IEEEmembership{Member,~IEEE}
        \vspace{-23pt}
        }

\markboth{Journal of \LaTeX\ Class Files,~2026}%
{Shell \MakeLowercase{\textit{et al.}}: Bare Demo of IEEEtran.cls for IEEE Journals}

\maketitle

\begin{abstract}
Efficient beam alignment is fundamental to high-throughput and reliable connectivity in Vehicle-to-Everything (V2X) systems. However, conventional beam management in dynamic vehicular topologies incurs prohibitive alignment overhead and struggles to maintain robust links under rapid mobility. To overcome these challenges, this paper proposes a distributed multimodal graph beam alignment (GBA) framework. The core innovation lies in leveraging onboard multimodal sensing data to predict implicit feedback while employing graph neural networks to coordinate multi-user alignment, thereby jointly enhancing scalability and drastically reducing overhead. The architecture adopts a dual-network design with GBA-RSU and GBA-Vehicle units, optimized through a hybrid strategy of centralized learning and federated learning (FL) to balance global performance with local privacy. Furthermore, a dedicated data augmentation (DA) scheme is introduced to address multimodal data imbalance issues in vehicular networks. Negative augmentation applies dominant modality dropout to bolster robustness, while positive augmentation generates underrepresented samples to mitigate label imbalance. Numerical results demonstrate that GBA maintains a competitive sum rate on par with high-resolution codebook-based feedback yet reduces beam alignment overhead by over $90\%$ and scales efficiently in mobile scenarios. Notably, integrating DA enables GBA to consistently outperform state-of-the-art FL-based alignment benchmarks, with particularly pronounced gains under severe label and modality imbalance, establishing a practical solution for V2X beam management.

\end{abstract}

\vspace{-3pt}

\vspace{-12pt}
\section{Introduction}

Millimeter-wave (mmWave) communication is vital to maintain high-throughput data transmission in Vehicle-to-Everything (V2X) systems \cite{CVGN-ComM-16}. To compensate for severe propagation loss, precise beamforming at the roadside unit (RSU) is essential \cite{RSPL-Com_M-14}. However, high mobility of vehicles shortens the channel coherence time, leaving a only limited time window for beam training and channel state information (CSI) acquisition \cite{GaCY-TWC-20}. As a result, conventional beam alignment methods relying on exhaustive beam initialization and downlink channel estimation (CE) incur prohibitive latency and signaling overhead. Consequently, aligning beam accurately with minimal overhead is critical for maintaining reliable and low-latency connectivity in V2X networks \cite{CDGY-IoTJ-22}.

Extensive research has sought to address this challenge. While precoding techniques like zero-forcing (ZF) \cite{JoUN-TSP-05} and weighted minimum mean square error (WMMSE) \cite{PeJa-TSP-22} perform well under perfect CSI, the expansion of antenna arrays leads to unsustainable computational complexity and pilot overhead. To mitigate these issues, recent studies have integrated deep learning (DL) into alignment optimization \cite{LCSZ-CommMag-19}. However, most DL-based methods rely on neural network (NN) pre-trained for static configurations. In dynamic V2X environments with varying user counts, such rigidity necessitates the training and storage of separate NNs for every possible configuration, resulting in excessive storage and computational demands. To enhance scalability, graph neural networks (GNNs) have been introduced due to their inherent adaptability to varying node counts \cite{SSZL-JSAC-21, LiYa-WCNC-24}. Nevertheless, these methods often still require extensive pilot sequences to match the performance of traditional optimization schemes.

Within the integrated sensing and communications (ISAC) framework, radio frequency (RF) sensing-aided beam alignment has emerged as a promising alternative to conventional methods \cite{WaBi-SPL-23, ShYH-TWC-25, XWA-TWC-25}. Building on this trend, the growing intelligence of modern vehicles has shifted attention toward multimodal sensing, which offers richer environmental context than RF-only solutions. Several studies have leveraged multimodal data to reduce overhead for beam selection \cite{ZGCY-TWC-24, ZWLH-TCOM-25} and beamforming \cite{NeCl-WCL-23, LiLY-TWC-25}. However, these approaches typically rely on RSU-mounted sensors, an infrastructure-centric design suffering from limited coverage \cite{ZGWC-TMC-25}. A natural extension is to exploit vehicle-mounted sensing to capture fine-grained environmental features directly at the user side. \cite{Sale-TVT-22} investigates vehicle-side sensing data collection for beam selection via centralized learning (CL), but it may raise privacy concerns given that multimodal data can encode richer semantic and identity-related cues \cite{CLZC-ACL-25}. Federated learning (FL) offers a privacy-preserving alternative \cite{SRGD-TMC-24, CSBC-ICC-25}. Yet existing FL-based methods assume balanced sensor configurations. It rarely holds in practical vehicular networks, where sensor setups are highly imbalanced. Training multimodal NN on such imbalanced data can lead to optimization divergence. Even when convergence occurs, unequal convergence rates across modality-specific branches may introduce bias, degrading both accuracy and robustness against missing modalities \cite{WHDW-TPAMI-25}. Compounding this issue, the label distribution across local datasets is often imbalanced, causing underrepresented labels to be diluted during global aggregation and further impairing prediction accuracy. In summary, designing a scalable and low-overhead beam alignment scheme that jointly addresses the dual challenge of modality and label imbalances in distributed multimodal sensing remains a critical and unresolved problem.

To address these research gaps, we propose a multimodal distributed graph beam alignment (GBA) scheme. Within this framework, each vehicle predicts low-overhead feedback using its multimodal sensing to minimize alignment overhead. The RSU then generates a scalable alignment strategy from the dynamic feedback graph to provide transmission service. To preserve privacy while capturing global network dynamics, the GBA framework is realized through two functional units: a GBA-Vehicle trained on private local datasets using FL, and a GBA-RSU trained on feedback graph via CL. These units are deployed on local vehicles and the infrastructure, respectively. To accommodate dynamic V2X networks, the GBA-RSU is built on a GNN architecture that trains vehicles as nodes and their interference relations as edges. This allows GBA-RSU to perform permutation-invariant processing and handle an arbitrary number of active links without requiring re-training for different network sizes. The GBA-Vehicle includes multiple dedicated branches to extract and fuse multimodal features for predicting local feedback with low overhead.

It is worth noting that the performance of this distributed framework is constrained by multimodal data imbalance. This imbalance manifests in two distinct forms: modality imbalance, which arises from heterogeneous sensor configurations and intermittent modality unavailability, and label imbalance, which stems from non-uniform vehicle trajectories within the RSU coverage area. To address these issues, we introduce a targeted augmentation (DA) strategy guided by three metrics: modality completeness, contribution ratio, and label overlap. The strategy comprises negative augmentation (DA-), which mitigates modality imbalance by dropping dominant modalities identified by the RSU to balance branch optimization, and positive data augmentation (DA+) , which alleviates label imbalance by generating synthetic samples for underrepresented labels using statistical information stored in the batch normalization (BN) layers of the globally aggregated model. This dual approach ensures robust and fair distributed training across diverse vehicles, enhancing the performance of GBA.

The key contributions of this work are as follows:

\begin{itemize}
	\item[ $\bullet$]
        We establish a distributed beam alignment framework that combines two complementary learning paradigms. On the RSU side, a GNN-based model is trained on feedback graphs to generate alignment strategies for varying vehicle counts and spatial distributions. In parallel, vehicle-side multimodal model is trained via FL to predict implicit feedback without exposing raw sensor data. Together, these components enable scalable and low-overhead beam management in dynamic V2X systems.
\end{itemize}

\begin{itemize}
	\item[ $\bullet$]
        We develop a targeted DA strategy that resolves the twin challenges of modality and label imbalances in multimodal FL. Guided by global statistics from the RSU, the strategy employs negative DA- to drop over-dominant modalities and prevent biased optimization, while DA+ synthesizes underrepresented samples to correct skewed label distributions. This dual approach ensures robust, fair distributed training and is, to our knowledge, the first augmentation scheme tailored to the imbalance patterns inherent in vehicular multimodal sensing.
\end{itemize}

\begin{itemize}
	\item[ $\bullet$]
        Extensive simulations on real-world driving datasets demonstrate that the proposed GBA reduces alignment overhead by $93.75\%$ compared to traditional optimization methods that rely on high-resolution codebook feedback, while maintaining competitive sum rates and high scalability. Moreover, this framework consistently outperforms state-of-the-art FL-based alignment benchmarks, with the performance margin widening substantially under severe imbalance, confirming the practical robustness.
\end{itemize}

\textit{Notations:} In this paper, we use lowercase letters $a$, bold lowercase letters $\mathbf{a}$, and bold uppercase letters $\mathbf{A}$ to denote scalars, vectors, and matrices respectively. The $i$-th element of a vector $\mathbf{a}$ is denoted by $[\mathbf{a}]_i$, while $[\mathbf{A}]_{i,j}$, $[\mathbf{A}]_{:, j}$, and $[\mathbf{A}]_{i, :}$ denote the $(i, j)$-th entry, $j$-th column, and $i$-th row of matrix $\mathbf{A}$. The $(\cdot)^{\top}$ represents the transpose operators. The operators $\Vert \cdot \Vert_{1}$ and $\Vert \cdot \Vert_{2}$ stand for the $\ell_1$-norm and $\ell_2$-norm. Furthermore, $\mathbb{E}[\cdot]$ and $\mathrm{Tr}(\cdot)$ represent the expectation and trace operations. The symbol $\oplus$, $\odot$, and $\otimes$ denote the concatenation, Hadamard product, and Exclusive OR (XOR) operations. Given vectors $\mathbf{a} \in \mathbb{C}^{M\times 1}$ and $\mathbf{b} \in \mathbb{C}^{N \times 1}$, $\mathbf{a} \oplus \mathbf{b} \in \mathbb{C}^{(M+N) \times 1}$. $\mathbf{1}_M$ and $\mathbf{0}_M$ denote the $M \times 1$ all-ones and all-zero vectors. Similarly, $\mathbf{I}_M$ represents the $M \times M$ identity matrix.

\vspace{-8pt}
\section{System Model and Problem Formulation}

As illustrated in Fig. \ref{fig1}, we consider a multi-user system where an RSU is equipped with an $N_t$-element uniform array to serve $K$ single-antenna users. Let $\mathbf{s} \in \mathbb{C}^{K \times 1}$ denote the vector of transmitted symbols, where $\mathbb{E}[\mathbf{s}\mathbf{s}^\text{H}] = \mathbf{I}_K$. The transmitted signal $\mathbf{x} \in \mathbb{C}^{N_t \times 1}$ is given by:
\begin{equation}
\mathbf{x} = \mathbf{W}\mathbf{P}^{1/2}\mathbf{s} = \sum_{k=1}^K \sqrt{p_k} \mathbf{w}_k s_k,
\label{eq1}
\end{equation}
where $\mathbf{w}_k \in \mathbb{C}^{N_t \times 1}$ is the $k$-th column of the precoding matrix $\mathbf{W}$, restricted to the unit sphere $\|\mathbf{w}_k\|_2 = 1$. The power allocation is governed by the diagonal matrix $\mathbf{P} = \text{diag}(p_1, \dots, p_K)$, subject to the total power constraint $\text{Tr}(\mathbf{P}) \leq P_{\text{max}}$. Since each precoding vector is selected from a predefined codebook $\mathbf{C} = [\mathbf{c}_1, ..., \mathbf{c}_W] \in \mathbb{C}^{N_t \times W}$ via a $W$-dimensional one-hot vector $\mathbf{u}_k \in \{0,1\}^{W \times 1}$, the precoding matrix is expressed as $\mathbf{W} = \mathbf{C}\mathbf{U}$, where $\mathbf{U} = [\mathbf{u}_1, \dots, \mathbf{u}_K] \in \{0,1\}^{W \times K}$ represents the beam selection matrix, satisfying $\mathbf{U}^\top\mathbf{1}_W = \mathbf{1}_K$.

Let $\mathbf{h}_k \in \mathbb{C}^{N_t \times 1}$ stand for the channel between the RSU and the $k$-th user. The received signal is expressed as $y_k = \mathbf{h}_k^H \mathbf{w}_k \sqrt{p_k} s_k + \sum_{i \neq k}^{K} \mathbf{h}_k^H\mathbf{w}_i \sqrt{p_i} s_i + n_k$, where $n_k \sim \mathcal{CN}(0, \sigma^2)$ is the additive white Gaussian noise. Accordingly, the achievable rate of the $k$-th user can be written as:
\begin{equation}
R_k=\log_2\!\Big(1 + \frac{p_k|\mathbf{h}_k^H \mathbf{w}_k|^2}{\sum_{i \neq k} p_i| \mathbf{h}_k^H\mathbf{w}_i|^2 + \sigma^2}\Big).
\label{eq2}
\end{equation}

\noindent To design $\mathbf{W}$ and $\mathbf{P}$, the RSU typically transmits reference signals to each user via defined beams in the codebook and then collects corresponding received signal strength (RSS) vectors for CE, which are denoted as $\mathbf{r}_k = \left[ |\mathbf{h}_k^H \mathbf{c}_1|^2, \ldots, |\mathbf{h}_k^H \mathbf{c}_W|^2 \right]^\top$. To mitigate the delay and feedback overhead incurred by this process, we propose that the user-$k$ utilizes a NN ${\mathcal{S}}(\cdot)$ with parameters $\mathbf{\Theta}_s$ to directly predict its feedback $\mathbf{v}_k$ based on the multi-modal sensing data $\mathbb{D}_k$. The feedback vector $\mathbf{v}_k \in \{0,1\}^W$ is generated via a binary quantization mapping $\mathcal{Q}: \mathbb{R}^W \to \{0,1\}^W$, such that the $i$-th entry is given by $[\mathbf{v}_k]_i = \mathbbm{1}([\mathbf{r}_k]_i > 0)$, where $\mathbbm{1}(\cdot)$ denotes the indicator function. It requires only $W$ bits to transmit, incurring much smaller overhead than transmitting $\mathbf{r}_k$. After collecting $\mathbf V = [\mathbf{v}_1, \dots, \mathbf{v}_K] \in \{0,1\}^{W \times K}$, the RSU employs a NN parameterized by $\mathbf{\Theta}_g$ to map the feedback graph directly to a feasible beamforming and power space, which is expressed as $\mathcal{G}(\cdot): \{0,1\}^{W \times K} \to \mathcal{U} \times \mathcal{P}$.

Under this model, the task of maximizing the total sum rates is formulated as:
\begin{subequations}
\begin{align}
\label{eq3}
\max_{\mathbf{\Theta}_s, \mathbf{\Theta}_g} \quad
& \sum_{k=1}^{K} R_k\\
\text{s.t.}\quad
& \mathbf{v}_k = {\mathcal{S}}(\mathbb{D}_k), \quad \forall k, \label{eq3b}\\
& (\mathbf{U}, \mathbf{P}) = {\mathcal{G}}(\mathbf V), \label{eq3c}\\
& \mathrm{Tr}(\mathbf{P}) \leq P_{\rm max}, \label{eq3d}\\
& \mathbf{U}^\top\mathbf{1}_W = \mathbf{1}_K. \label{eq3e}
\end{align}
\end{subequations}
Let $\mathcal{Q}$ be the set of available sensing modalities. For each vehicle within the RSU's coverage $k \in \mathcal{K}$, the local sensing profile is defined by the subset $\mathcal{Q}_k \subseteq \mathcal{Q}$. The observed data is represented as a multimodal tuple $\mathbb{D}_k = \{ \mathbf{x}_k^{(q)} \}_{q \in \mathcal{Q}_k}$, where each $\mathbf{x}_k^{(q)}$ resides in a modality-specific Hilbert space $\mathcal{H}_q$. For notation simplicity, we use $\mathbf{T} \triangleq \mathbf{UP}^{1/2}$.

\vspace{-10pt}
\section{Distributed Graph Beam Alignment Scheme}
In this section, we introduce the proposed GBA scheme, comprising GBA-RSU and GBA-Vehicle. The tailored training and execution procedures of GBA are presented as follows.

\vspace{-10pt}
\subsection{GBA-Vehicle Design}

To model the relationship between vehicle-side multi-modal sensing data and the corresponding feedback, we propose GBA-Vehicle. This architecture comprises modality-specific extraction, along with a fusion branch. Specifically, a residual network (ResNet) based branch $\mathcal{S}^{\rm R}(\cdot)$ extracts multipath features from RGB images; a convolutional neural network (CNN)-based branch $\mathcal{S}^{\rm L}(\cdot)$ processes the pre-processed 3D LiDAR data\footnote{The original point clouds are partitioned into cuboidal regions based on distance from the RSU, and each region is assigned a unique label. Blocking obstacles are denoted by $1$, the BS by $-1$, and the receiver by $-2$\cite{SRGD-TMC-24}.} to capture the spatial structural information of the communication environment; and another CNN-based branch $\mathcal{S}^{\rm G}(\cdot)$ extracts line-of-sight (LoS) geometric features from the GPS data. Let $\mathbf{f}_k^{\rm R} \in \mathbb{R}^{L_{\rm R}}$, $\mathbf{f}_k^{\rm L} \in \mathbb{R}^{L_{\rm L}}$, $\mathbf{f}_k^{\rm G} \in \mathbb{R}^{L_{\rm G}}$ denote the RGB, LiDAR, and GPS features of the $k$-th vehicle, respectively. The extraction processes are expressed as:
\begin{equation}
    \label{eq11_GBAv}
    \begin{alignedat}{3}
    \mathbf{f}_k^{\rm G} &{}= \mathcal{S}^{\rm G}(\mathbf{x}^{\rm G}_k), &\quad& \mathcal{S}^{\rm G}: &\; & \mathbb{R}^{2} \mapsto \mathbb{R}^{L_{\rm G} \times 1},\\
    \mathbf{f}_k^{\rm R} &{}= \mathcal{S}^{\rm R}(\mathbf{x}^{\rm R}_k), &\quad& \mathcal{S}^{\rm R}: &\; & \mathbb{R}^{d_0^{\rm R}\times d_1^{\rm R}\times d_2^{\rm R}} \mapsto \mathbb{R}^{L_{\rm R} \times 1},\\
    \mathbf{f}_k^{\rm L} &{}= \mathcal{S}^{\rm L}(\mathbf{x}^{\rm L}_k), &\quad& \mathcal{S}^{\rm L}: &\; & \mathbb{R}^{d_0^{\rm L}\times d_1^{\rm L}\times d_2^{\rm L}} \mapsto \mathbb{R}^{L_{\rm L} \times 1}.
    \end{alignedat}
\end{equation}
where $(d_0^{\text{R}} \times d_1^{\text{R}} \times d_2^{\text{R}})$ and $(d_0^{\text{L}} \times d_1^{\text{L}} \times d_2^{\text{L}})$ give the dimensions of the RGB and the preprocessed LiDAR data. Then, they are fused by $\mathbf{f}_k^{{\rm F}} = \mathbf{f}_k^{{\rm G}} \oplus \mathbf{f}_k^{{\rm R}} \oplus \mathbf{f}_k^{{\rm L}}$, which is fed into a multi-layer perceptron (MLP)-based branch $\mathcal{S}^{\rm F}(\cdot)$ to predict the feedback:
\begin{equation}
\label{eq46_GBAv}
    \hat{\mathbf{v}}_k= \mathcal{S}^{\rm F}(\mathbf{f}_k^{\rm F}), \quad \mathcal{S}^{\rm F}: \mathbb{R}^{(L_{\rm G}+L_{\rm R}+L_{\rm L}) \times 1} \mapsto \mathbb{R}^{W \times 1}.
\end{equation}
Then, the total feedback vectors are transmitted to the RSU, where an alignment strategy is generated by the GBA-RSU.

\begin{figure}[!t]
    \centering
    \includegraphics[width=0.85\columnwidth]{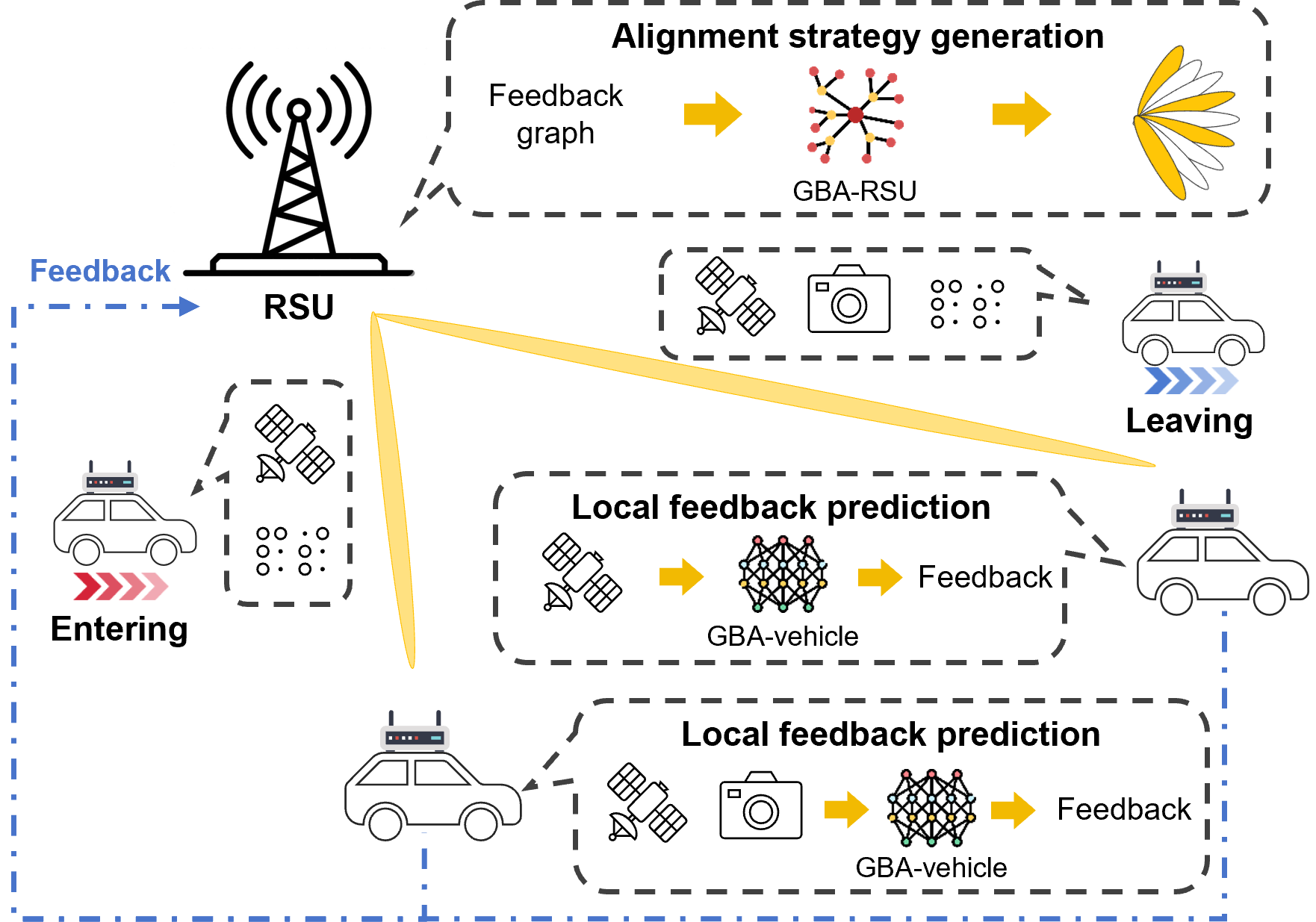}
    \caption{An illustration of the system model.}
    \label{fig1}
\end{figure}

\vspace{-10pt}
\subsection{GBA-RSU Design}

To map the feedback graph to the optimal pair $(\mathbf{U}, \mathbf{P})$ for varying numbers of vehicles, we design a GNN-based GBA-RSU. We begin by proving in the following proposition that the alignment strategy generation problem formulated in Eq. \eqref{eq3c}, subject to the constraints in Eqs. \eqref{eq3d}-\eqref{eq3e}, is invariant under any permutation of the vehicle indices.

\begin{proposition}
\label{proposition1}
\textit{[Permutation Invariance]:}
Let $\Pi \in \{0, 1\}^{K \times K}$ be an arbitrary permutation matrix. The feasible set of the beam alignment problem defined in (3c) is invariant under the transformation $(\mathbf{U}, \mathbf{P}) \to (\mathbf{U}\Pi, \Pi^\top \mathbf{P} \Pi)$. Furthermore, the total sum-rate objective function preserves this invariance, such that:
\begin{equation}
\sum_{k=1}^{K}R_k(\mathbf{U}\Pi,\Pi^\top \mathbf{P} \Pi) = \sum_{k=1}^{K} R_k(\mathbf{U}, \mathbf{P})
\end{equation}
\end{proposition}

\emph{Proof}: See Appendix \ref{P1}.

\noindent This invariance implies that the optimal output should transform consistently with any permutation of vehicles. Specifically, when the input is permuted as $\mathbf{V}\mathbf{\Pi}$, the optimal pair should transform as $(\mathbf{U}^\star\mathbf{\Pi}, \mathbf{\Pi}^\top\mathbf{P}^\star\mathbf{\Pi})$, and the corresponding output becomes $\mathbf{T}^\star\mathbf{\Pi}=\mathbf{U}^\star\mathbf{P}^\star\mathbf{\Pi}$. Hence, the GBA-RSU $\mathcal{G}(\cdot)$ should satisfy $\mathbf{T}^\star\mathbf{\Pi}=\mathcal{G}(\mathbf{V}\mathbf{\Pi})$ to handle an arbitrary number of active links. Accordingly, it can be designed as a GNN \cite{WYHW-TMLCN-24}.

To design the GBA-RSU, a graph is consturcted for dynamic V2X network, where the interference relationships among vehicles change dynamically due to mobility, entering and leaving. Specifically, let $\mathcal Z=(\mathcal K,\mathcal E)$ denote the dynamic graph. The vertex set is given by the vehicle set $\mathcal{K}$, and the feature of vertex $k$ is denoted by $\mathbf{v}_k$. The edge set $\mathcal E$ represents the interference links formed when two vehicles share at least one beam candidate. The adjacency matrix is denoted as:
\begin{equation}
\label{eq6}
[\mathbf{A}]_{ij} =
\begin{cases}
1, & [\mathbf{\mathbf{V}^\top\mathbf{V}}]_{ij} \ge 1 \text{ and } i \neq j,\\
0, & \text{otherwise},
\end{cases}
\end{equation}
Since interference is bidirectional, all edges within the graph $\mathcal Z$ are undirected. The feature associated with edge $(i,j)$ is defined by concatenating the features of its two vertices, denoted by $\mathbf{v}_{i,j} = \mathbf{v}_i \oplus \mathbf{v}_j \in \mathbb{R}^{2W}$.

The constructed graph is then fed into the GBA-RSU for alignment strategy generation, as illustrated in Fig. \ref{fig2}. The GBA-RUS comprises three components: an edge encoding layer, a vertex encoding layer, and a beam projection layer. The details are described as follows.

\begin{figure*}[!t]
    \centering
    \includegraphics[width=0.85\textwidth]{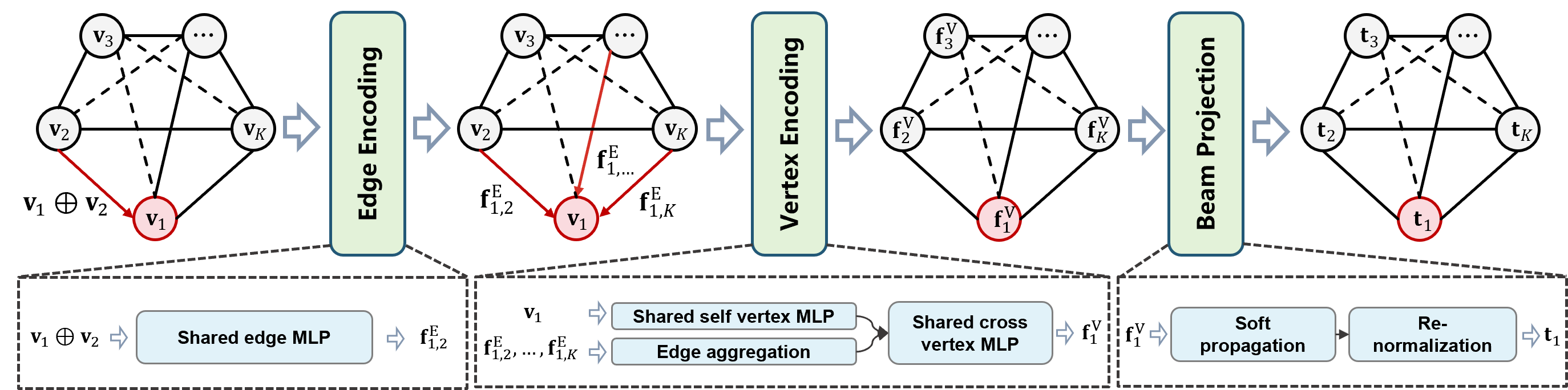}
    \caption{The structure of the GBA-RSU.}
    \label{fig2}
    \vspace{-18pt}
\end{figure*}

\subsubsection{\textbf{Edge Encoding Layer}}
For the $k$-th vertex, the neighbor edge set is constructed by $\mathcal{E}_k=\{e_{k,1},..., e_{k,K}\}$, where $e_{k,i}=1$ indicates the existence of an edge between vertex-$k$ and vertex-$i$. The interference information of the $k$-th vertex is extracted by the an edge MLP, yielding $\mathbf{f}^{\mathrm{E}}_{k,i}=\mathcal{G}^{\mathrm{E}}(\mathbf{v}_{k,i}) \in \mathbb{R}^{d_{\mathrm{g}}}$, where $d_{\mathrm{g}}$ denotes the hidden dimension of the GBA-RSU. The egde MLP is shared across all edges to ensure scalability.

\subsubsection{\textbf{Vertex Encoding Layer}}
At this layer, the extracted interference information of $k$-th vertex is aggregated as:
\begin{equation}
\label{eq7}
\overline{\mathbf{f}^{\text{E}}_k}=\frac{\sum_{(k,i)\in \mathcal{E}_k}\mathbf{f}^{\text{E}}_{k,i}}{|\mathcal{E}_k|}.
\end{equation}
Simultaneously, the intrinsic vertex feature is processed by a shared self-vertex MLP, denoted as $\mathbf{f}_k^{\text{SV}}=\mathcal{G}^{\rm SV}(\mathbf v_k) \in \mathbb{R}^{d_{\rm g}}$, which is then concatenated with the aggregated edge information to form $\mathbf{f}_k^{\mathrm{EV}} =\mathbf{f}_k^{\text{SV}} \oplus \overline{\mathbf{f}^{\mathrm{E}}_k}$. Finally, the vertex representation is updated via $\mathbf{f}^{\mathrm{V}}_k=\mathcal{G}^{\mathrm{CV}}(\mathbf{f}_k^{\mathrm{EV}})$, where $\mathcal{G}^{\mathrm{CV}}(\cdot)$ represents the shared cross-vertex MLP.

\subsubsection{\textbf{Beam Projection Layer}}

The beam projection layer is tailored for mapping the updated vertex features to a feasible alignment strategy $\mathbf{T}=\mathbf{UP}^{1/2}$. At this layer, each vertex feature is projected back to the $W$-dimensional vector via $\mathbf{z}_k^{\rm V} = \mathcal{G}^{\rm V}_{\rm I}(\mathbf{f}^{\text{V}}_k)$. However, these results violate the constraints in Eq.~\eqref{eq3d} and Eq.~\eqref{eq3e} since they are multi-hot vectors with some negative elements. To address this problem, the soft propagation and re-normalization modules are created. Specifically, for soft propagation module, $\mathbf{z}_k^{\rm V}$ is enforced to be non-negative by $\mathbf{z}_k^{\text{non}} = (\mathbf{z}^{\mathrm{V}}_k \odot \mathbf{z}^{\mathrm{V}}_k)^{1/2}$. Then, the hard mask $\mathbf{m}_k^{\text{hard}}$ is generated via an argmax operation on $\mathbf{m}_k^{\text{hard}}$ and the non-negative one-hot output is obtained via $\mathbf{z}_k^{\rm hard} = \mathbf{z}_k^{\text{non}} \odot \mathbf{m}_k^{\text{hard}}$, satisfying the constraints in Eq.~\eqref{eq3e}. However, the above mask operation is non-differentiable and thus breaks the gradient backpropagation during NN training. To solve this problem, we make a differentiable approximation of $\mathbf{m}_k^{\text{hard}}$ by introducing a soft mask via the softmax operation:
\begin{equation}
\label{eq8}
    [\mathbf{m}_k^{\text{soft}}]_i = \frac{\exp({\frac{[\mathbf{z}_k^{\text{non}}]_i}{\tau}})}{\sum_{j=1}^{W} \exp({\frac{[\mathbf{z}_k^{\text{non}}]_j}{\tau}})}, \quad i = 1,2,\ldots,W,
\end{equation}
where $\tau$ is a temperature parameter controlling the smoothness of selection. This approximation is theoretically justified by the asymptotic behavior of the softmax operator. Assuming that the maximum entry of $\mathbf{z}_k^{\text{non}}$ is unique, $\mathbf{m}_k^{\text{soft}}$ converges to $\mathbf{m}_k^{\text{hard}}$ as $\tau \to 0^+$. In particular, after normalizing Eq.~\eqref{eq8} by the exponential term associated with the maximum entry, the dominant term remains $1$, while all other terms vanish due to strictly negative exponents. Assisted by such soft mask, the mixed mask $\mathbf{m}_k^{\text{mix}}$ is created by $\mathbf{m}_k^{\text{mix}} = (\mathbf{m}_k^{\text{hard}} - \mathbf{m}_k^{\text{soft}})_{\text{detach}} + \mathbf{m}_k^{\text{soft}}$ to separate non-differentiable forward propagation and differentiable backward propagation\footnote{In the forward propagation, $\mathbf{m}^{\text{mix}}$ is numerically equivalent to $\mathbf{m}_k^{\text{hard}}$. In the backward propagation, the gradients are computed via $\mathbf{m}^{\text{soft}}_k$, enabling differentiability.}, where $(\cdot)_{\text{detach}}$ denotes the stop-gradient operation. Currently, the output is expressed as $\hat{\mathbf{z}}_k=\mathbf{z}_k^{\text{hard}} \odot \mathbf{m}_k^{\text{mix}}$, which complies with the constraint in Eq.~\eqref{eq3e}. The re-normalization module is applied on $\hat{\mathbf{z}}_k$ to satisfy the constraint in Eq.~\eqref{eq3d}. Specifically, we treat $\|\hat{\mathbf{z}}_k\|_2$ as the initial transmit power of the $k$-th vehicle, and dynamically prune vertices with power occupancy below $0.05\%$, which typically provide limited sum rates gains while introducing additional interference into the system. The remaining vectors are then normalized as the predicted alignment strategy by:
\begin{equation}
\hat{\mathbf{t}}_k = \sqrt{\frac{P_{\text{max}}}{\sum_{j=1}^{K} \|\hat{\mathbf{z}}_j\|_2^2}}~\hat{\mathbf{z}}_k,
\end{equation}
which satisfies the constrains in Eq.~\eqref{eq3d} and Eq.~\eqref{eq3e}.

Recalling that the GBA-RSU is required to satisfy permutation equivariance with respect to the vertex ordering, we present the following proposition.
\begin{proposition}
\label{Proposition2}
\textit{[Permutation Equivariance]:}
For any input feedback matrix $\mathbf{V} \in \mathbb{R}^{W \times K}$ and any permutation matrix $\Pi \in \{0, 1\}^{K \times K}$, the mapping $\mathcal{G}(\cdot)$ defined by GBA-RSU is permutation equivariant. Specifically, the output alignment strategy satisfies:
\begin{equation}
\mathcal{G}(\mathbf{V}\mathbf{\Pi}) = \mathcal{G}(\mathbf{V})\mathbf{\Pi}.
\end{equation}
\end{proposition}
This equivariance holds for each each layer of GBA-RSU, including edge encoding, vertex encoding, and beam projection.

\emph{Proof}: See Appendix \ref{P2}.

\textit{Remark 1:} Unlike MLPs with a fixed input dimension, the proposed GBA-RSU architecture leverages the shared feature processing and permutation equivariance of GNNs. This ensures that the RSU generalizes well to different numbers of vehicles and maintains consistent alignment performance regardless of vehicle indexing order.


\vspace{-8pt}
\subsection{Offline Training Procedure of GBA scheme}
The offline training procedure of GBA comprises three stages, which are detailed below.

\textbf{[Stage 1] GBA-RSU training:} The RSU trains GBA-RSU using the ground truth feedback. To improve its scalability to varying network topologies, a stochastic topology perturbation strategy is adopted during training to emulate the random arrival and departure of vehicles. Specifically, in the $j$-th iteration, let $\mathbf D \in \{0,1\}^{W \times K}$ denote the binary dropping matrix, where each entry $[\mathbf{D}]_{:,k}$ takes the value of $\mathbf{0}_W$ with a dropping probability of $p_{\text{drop}}$. The perturbed feedback matrix of vehicles is represented as $\tilde{\mathbf V}=\mathbf D \odot \mathbf V$. Accordingly, the RSU constructs the adjacency matrix $\tilde{\mathbf A}$ via Eq.~\eqref{eq6}. The beam alignment strategy is generated by $\hat{\mathbf{T}} = \mathcal{G}(\tilde{\mathbf A}, \tilde{\mathbf V})$ and then used to calculate the loss function $\mathcal{L_{\rm GBA\text{-}R}}$ denoted by Eq.~\eqref{eq3}. The parameters are updated by $\mathbf{\Theta}_g^j=\mathbf{\Theta}_g^{j-1}-\eta_g\nabla_{\mathbf{\Theta}_g^j} \mathcal{L_{\rm GBA\text{-}R}}$, where $\eta_g$ is the learning rate of GBA-RSU training.

\begin{algorithm}[t!]
\caption{Offline Training Procedure of GBA Scheme}
\label{alg1}
\underline{\textbf{Input:}} Local multi-modal dataset of vehicles $\{\mathbb{D}_k\}_{k=1}^K$, \\
\phantom{\underline{\textbf{Input:}}} Initialized $\mathbf{\Theta}_g$ and $\{\mathbf{\Theta}^{\text{F}}_{s,k}, \{\mathbf{\Theta}_s^{q}\}_{q \in \mathcal{Q}_k}\}_{k=1}^K$. \\
\underline{\textbf{Output:}} Trained GBA-RSU $\mathbf{\Theta}_g$, trained GBA-Vehicle $\mathbf{\Theta}_s$.
\begin{algorithmic}[1]
\Statex // \textit{Stage 1: GBA-RSU training stage}
\For {$j = 1,..., N_{\text{tra}}$} \textbf{do}
  \State $[\mathbf{D}]_{:,k} \sim \text{Bernoulli}(p_{\text{drop}})$.
  \State $\tilde{\mathbf V} \gets \mathbf D \odot\mathbf V$ and extract $\tilde{\mathbf{A}}$ via Eq.~\eqref{eq6}.
  \State $\hat{\mathbf{T}} \gets \mathcal{G}(\tilde{\mathbf{A}}, \tilde{\mathbf{V}}; \mathbf{\Theta}_g^j)$;
  \State $\nabla_{\mathbf{\Theta}_g^j} \mathcal{L}_{\text{GBA-R}}^j \gets \frac{\partial \mathcal{L}_{\text{GBA-R}}^j}{\partial \tilde{\mathbf V}}\frac{\partial \tilde{\mathbf V}}{\partial \mathbf{\Theta}_g^j}$;
  \State $\mathbf{\Theta}_g^j\gets\mathbf{\Theta}_g^{j-1}-\eta_g\nabla_{\mathbf{\Theta}_g^j} \mathcal{L}_{\text{GBA-R}}^j$;
\EndFor

\Statex // \textit{Stage 2: Distributed GBA-Vehicle training}
\For {$t = 1,..., N_{\text{glob}}$} \textbf{do}
  \For {\textbf{all} vehicle $k$ \textbf{in parallel}}
    \For {$r \gets 1,..., N_{\text{loc}}$} \textbf{do}
            \State $\hat{\mathbf{v}}_k^r \gets \mathcal{S}_k(\mathbb{D}_k;\mathbf{\Theta}_{s,k}^{r,t})$;
            \State $\nabla_{\mathbf{\Theta}_{s,k}^{r,t}} \mathcal{L}_{{\rm BCE},k}^{r,t} \gets \frac{\partial \mathcal{L}_{{\rm BCE},k}^{r,t}}{\partial \hat{\mathbf{v}}_k^r} \frac{\partial \hat{\mathbf{v}}_k^r}{\partial \mathbf{\Theta}_{s,k}^{r,t}}$;
            \State $\mathbf{\Theta}^{\text{F},r+1,t}_{s,k} \gets \mathbf{\Theta}^{\text{F},r,t}_{s,k}- \eta_s\,\nabla_{\mathbf{\Theta}^{\text{F},r,t}_{s,k}}\mathcal{L}_{{\rm BCE},k}^{r,t}$;
            \State $\{\mathbf{\Theta}^{{q},r+1,t}_{s,k} \!\gets\! \mathbf{\Theta}^{{q},r,t}_{s,k}- \eta_s\,\nabla_{\mathbf{\Theta}^{{q},r,t}_{s,k}}\mathcal{L}_{{\rm BCE},k}^{r,t}\}_{q \in \mathcal{Q}_k}$;
        \EndFor
    \State Send $\mathbf{\Theta}^{\text{F},t+1}_{s,k}$ and $\{\mathbf{\Theta}_s^{q,t+1}\}_{q \in \mathcal{Q}_k}$ to the RSU.
    \EndFor
    \State $\{\mathbf{\Theta}_s^{{q},t+1} \gets \frac{1}{\lvert \mathcal{K}_{q} \rvert} \sum_{k=1}^{|\mathcal{K}_{q}|} \mathbf{\Theta}^{{q},t}_{s,k}\}_{q \in \mathcal{Q}}$;
    \State $\mathbf{\Theta}^{\text{F},t+1}_{s} \gets \frac{1}{K} \sum_{k=1}^K \mathbf{\Theta}^{\text{F},t}_{s,k}.$
  \State Broadcast $\{\mathbf{\Theta}_s^{q,t+1}\}_{q \in \mathcal{Q}}$ and $\mathbf{\Theta}^{\text{F},t+1}_{s}$ to each vehicle.
\EndFor

\Statex // \textit{Stage 3: GBA-RSU retraining stage}
\State $\{\hat{\mathbf{v}}_k\}_{k=1}^K \gets \{\mathcal{S}(\mathbb{D}_k)\}_{k=1}^K$ and $\{\hat{\mathbf{v}}_k\}_{k=1}^K \cup \{\mathbf{v}_k\}_{k=1}^K$.
\For {$j = 1,..., N_{\text{r-tra}}$} \textbf{do}
  \State $[\mathbf{D}]_{:,k} \sim \text{Bernoulli}(p_{\text{drop}})$, $[\mathbf{E}]_{w,k} \sim \text{Bernoulli}(p_{\text{error}})$
  \State $\breve{\mathbf{V}} \gets \mathbf D \odot \mathbf{V}^{\text{mix}} \otimes \mathbf{E}$ and extract $\breve{\mathbf{A}}$ via Eq.~\eqref{eq6}.
  \State $\hat{\mathbf{T}} \gets \mathcal{G}(\breve{\mathbf{A}}, \breve{\mathbf{V}}; \mathbf{\Theta}_g^j)$;
  \State $\nabla_{\mathbf{\Theta}_g^j} \mathcal{L}_{\rm GBA\text{-}R}^j \gets \frac{\partial \mathcal{L}_{\rm GBA\text{-}R}^j}{\partial \breve{\mathbf V}}\frac{\partial\breve{\mathbf V}}{\partial \mathbf{\Theta}_g^j}$;
  \State $\mathbf{\Theta}_g^j\gets\mathbf{\Theta}_g^{j-1}-\eta_{g\text{-}r}\nabla_{\mathbf{\Theta}_g^j} \mathcal{L}_{\rm GBA\text{-}R}^j$;
\EndFor
\end{algorithmic}
\end{algorithm}

\textbf{[Stage 2] Distributed training of GBA-Vehicle:} Vehicles jointly train GBA-Vehicle by FL. At the $t$-th communication round, each vehicle performs local training using its local multi-modal sensing data. Due to heterogeneous sensor configurations across vehicles, only the extraction branches corresponding to the available modalities and the fusion branch are updated. Specifically, let $\mathbf{\Theta}^{\text{F}}_{s,k}$ and $\mathbf{\Theta}^{q}_{s,k}$ denote the parameters of the fusion branch and the modality $q$ extraction branch for $k$-th vehicle, respectively. The update processes at the $r$-th local training step are given by
$
\mathbf{\Theta}^{\text{F},r+1}_{s,k}
= \mathbf{\Theta}^{\text{F},r}_{s,k}
   - \eta_s\,
     \nabla_{\mathbf{\Theta}^{\text{F},r}_{s,k}}\mathcal{L}_{{\rm BCE},k}
$
and 
$
\mathbf{\Theta}^{q,r+1}_{s,k}
= \mathbf{\Theta}^{q,r}_{s,k}
   - \eta_s\,\nabla_{\mathbf{\Theta}^{q,r}_{s,k}}\mathcal{L}_{{\rm BCE},k}
$,

where $\eta_s$ is the learning rate. $\mathcal{L}_{{\rm BCE},k}$ denotes the binary cross-entropy (BCE) loss function for feedback prediction:
\vspace{-5pt}
\begin{equation}
\begin{aligned}
\mathcal{L}_{{\rm BCE},k}(\mathbf{v}_k, \hat{\mathbf{v}}_k)
= -&\sum_{w=1}^{W} \Big( [\mathbf{v}_k]_w\log [\hat{\mathbf{v}}_k]_w \\
&+ (1 - [\mathbf{v}_k]_w)\log\big(1 - [\hat{\mathbf{v}}_k]_w\big) \Big).
\label{eq14_GBA_tmp}
\end{aligned}
\end{equation}
Here, $\mathbf{v}_k$ and $\hat{\mathbf{v}}_k$ are the ground truth and predicted feedback vectors, respectively. After completing the local updates, each vehicle uploads the resulting parameters to the RSU for modality-aware aggregation, denoted by:
\begin{equation}
\mathbf{\Theta}^{q,t+1}_s = \frac{1}{\lvert \mathcal{K}_{q} \rvert} \sum_{k=1}^{|\mathcal{K}_{q}|} \mathbf{\Theta}^{q,t+1}_{s,k},
\quad
\mathbf{\Theta}^{\text{F},t+1}_{s} = \frac{1}{K} \sum_{k=1}^K \mathbf{\Theta}^{\text{F},t+1}_{s,k},
\label{eq19_GBA_tmp}
\end{equation}
where $\mathcal{K}_{q}$ is the set of vehicles equipped with modality $q$. Then, the aggregated parameters are broadcast to all vehicles for the next local training round. After convergence, the trained GBA-Vehicle is deployed on vehicles for feedback prediction.

\textbf{[Stage 3] GBA-RSU retraining:} Since inevitable errors in the predicted feedback may result in incorrect beam candidates for some vehicles, GBA-RSU is misled in generating alignment strategies. To mitigate this impact, the trained GBA-Vehicle is used to produce a predicted feedback dataset $\{\hat{\mathbf{v}}_k\}_{k=1}^K = \{\mathcal{S}(\mathbb{D}_k)\}_{k=1}^K$ to create the mxied dataset $\{\hat{\mathbf{v}}_k\}_{k=1}^K \cup \{\mathbf{v}_k\}_{k=1}^K$ for GBA-RSU retraining. Furthermore, to enhance the robustness against beam prediction errors, a augmentation strategy is used in this stage. Specifically, let $\mathbf E \in \{0,1\}^{W \times K}$ denote a binary matrix where each element takes the value of $1$ with a probability of $p_{\text{error}}$, serving to inject random beam errors. Since the stochastic topology perturbation strategy still remains to emulate the dynamic system, the augmented input is expressed as $\breve{\mathbf{V}} = \mathbf D \odot \mathbf{V}^{\text{mix}} \otimes \mathbf{E}$. The retraining of GBA-RSU continues to utilize the loss function defined in Eq.~\eqref{eq3}. The parameters are updated by $\mathbf{\Theta}_g^j=\mathbf{\Theta}_g^{j-1}-\eta_{g\text{-}r}\nabla_{\mathbf{\Theta}_g^j} \mathcal{L_{\rm GBA\text{-}R}}$, where $\eta_{g\text{-}r}$ is the learning rate of GBA-RSU retraining. 

Owing to the above offline training procedure summarized in Algorithm \ref{alg1}, the error accumulation caused by the separate deployment of GBA-RSU and GBA-Vehicle can be alleviated, thus enhancing their mutual cooperation.

\vspace{-10pt}
\subsection{Online Execution Strategy of GBA Scheme}

To mitigate online performance degradation caused by heterogeneous sensor configurations and intermittent modality unavailability, we design a modality-aware pruned NN deployment strategy. Upon initialization, each vehicle registers its sensor configuration with the RSU, which then assigns the extraction branches and pruned fusion branch corresponding to available modalities to vehicle. Specifically, let $(\mathbf{W}_{\rm F}, \mathbf{b}_{\rm F})$ denote the parameters of the fusion branch in GBA-Vehicle. The feedback prediction of the $k$-th vehicle is denoted as:
\begin{equation}
\mathbf{v}_k
=\mathbf{W}_{\rm F}\mathbf{f}_k^{{\text{F}}}+\mathbf{b}_{\rm F} = \sum_{q \in \mathcal{Q}_k} \mathbf{W}_{\rm F}^{q}\mathbf{f}_k^{{q}}+\mathbf{b}_{\rm F}^{q},
\end{equation}
where $(\mathbf{W}_{\rm F}^{q}, \mathbf{b}_{\rm F}^{q})$ is the parameter block associated with modality $q \in \mathcal{Q}_k$, and $\mathbf{f}_k^{q}$ is the corresponding extracted feature. This strategy ensures that there is no noise introduced by data padding to compensate for missing modalities, while reducing the communication overhead of model distribution. Moreover, if a modality fails, the RSU can dynamically redeploy the adapted components to maintain robustness.

Next, we analyze the computational complexity of the GBA scheme in the online phase, starting from the GBA-Vehicle unit. Let $K_{\rm e}$ and $C_{\rm o}$ denote the kernel size and output channels of CNN. The extraction branches exhibit similar computational complexity, denoted as $\mathcal{O}(C_{\rm o}^{\rm G}K_{\rm e}^2)$, $\mathcal{O}(d_0^{\rm R}d_1^{\rm R}C_{\rm o}^{\rm R}K_{\rm e}^2)$ and $\mathcal{O}(d_0^{\rm L}d_1^{\rm L}d_2^{\rm L}C_{\rm o}K_{\rm e}^2)$, respectively. In contrast, the fusion branch leads to a higher lower of $\mathcal{O}((L_{\rm G}+L_{\rm R}+L_{\rm L})W)$. Consequently, the overall complexity of the GBA-Vehicle is given by $\mathcal{O}(C_{\rm o}^{\rm G}K_{\rm e}^2 + d_0^{\rm R}d_1^{\rm R}C_{\rm o}^{\rm R}K_{\rm e}^2 + d_0^{\rm L}d_1^{\rm L}d_2^{\rm L}C_{\rm o}^{\rm L}K_{\rm e}^2)$. Regarding the proposed GBA-RSU, constructing the adjacency matrix entails a complexity of $\mathcal{O}(K^2W)$ due to the interference graph calculation. In the edge encoding layer, the edge MLP incurs a complexity of $\mathcal{O}(\sum_k^K|\mathcal{E}_k|Wd_{\rm g})$. The complexity of the vertex encoding layer is determined by neighbor aggregation and vertex self/cross MLPs, scaling as $\mathcal{O}(\sum_k^K|\mathcal{E}_k|d_{\rm g})$ and $\mathcal{O}(KWd_{\rm g}+Kd_{\rm g}^2)$, respectively. Additionally, the beam prediction layer contributes $\mathcal{O}(Kd_{\rm g})$. Finally, the computational complexity of GBA-RSU is dominated by $\mathcal{O}(K^2W)$, since the dimension $d_{\rm g}$ and the codebook size $W$ are fixed, and $\sum_k^K|\mathcal{E}_k| \ll K^2$ in practice. In summary, considering the kernel size $K_{\rm e}$ is relatively small, the total computational complexity of GBA is dominated by $\mathcal{O}(K^2W)$. Thus, the computational overhead satisfies real-time requirements.

\vspace{-7pt}
\section{Multimodal Data Augmentation Strategy}

This section first quantifies label and modality imbalances using three metrics. Then, the multimodal data DA strategy, comprising DA- and DA+, is presented. It mitigates modality imbalance through target modality dropout and alleviates label imbalance by data generation. Finally, a tailored coordinated application procedure is detailed.

\vspace{-10pt}
\subsection{Quantification of Multi-modal Data Imbalances}

Modality imbalance arises from heterogeneous sensor configurations across vehicles. Therefore, we use the modality completeness rate to quantify the volume discrepancy among different modalities. When training on such imbalanced data, the branch corresponding to modality with higher completeness rate dominates the gradient updates, thereby suppressing the optimization of other branches. Thus, we give the modality contribution ratio to evaluate the balance of updates among branches. Let $N_k^{q}$ denote the number of samples of modality $q$ on $k$-th vehicle. Both metrics are described as follows.

\textbf{Modality Completeness Rate}: Let $N_k^{q^\star}$ represent the maximum sample count among all modalities for $k$-th vehicle, the modality completeness rate is described as $\kappa^q_k = \frac{N_{k}^q}{N_{k}^{q^\star}}$.
\vspace{-5pt}
\begin{equation}
\kappa^q_k = \frac{N_{k}^q}{N_{k}^{q^\star}}, \quad 0 \le \kappa \le 1.
\label{equ4}
\end{equation}
The completeness rate equals $1$ if the corresponding sensor is available. When partial data are lost due to sensor malfunction, the rate falls between $0$ and $1$. Accordingly, the global modality completeness rate is expressed as $\kappa^q = \frac{1}{K}\sum_{k=1}^K\kappa^q_k$.

\textbf{Modality Contribution Ratio}: For each modality in $\mathcal{Q}_k$, the uni-modality based feedback is predicted by $\hat{\mathbf{v}}^{q}_k=\mathbf{W}_{\rm F}\mathbf{f}_k^{q}+\mathbf{b}_{\rm F}$ when the corresponding sensor is available. Based on it, GBA-RSU can generate the alignment strategy $\hat{\mathbf{t}}^{q}_k$ and can get achievable rate $R^{q}_k$. If the corresponding sensor is unavailable, the obtained achievable rate is set to $0$. Thus, the modality contribution ratio of $i$-th modality in $\mathcal{Q}_k$ is given by:
\begin{equation}
\varphi^{i}_k = \frac{1}{|\mathcal{Q}_k|-1} \sum_{i=1, j \neq i}^{|\mathcal{Q}_k|}\frac{R^{i}_k}{R^{j}_k}.
\label{eq5}
\end{equation}
If each $\varphi^{i}_k$ is always $1$ during training, the gradients are distributed uniformly across all branches. Otherwise, the gradients are biased towards the branch with higher contribution ratio, degrading the robustness of GBA-Vehicle.

Label imbalance arises from different trajectories of vehicles. Training on such data leads to biased global aggregation, where the trained GBA-Vehicle overlooks features associated with underrepresented labels. To quantify this impact, we introduce the label overlap proportion. Let $M=2^W$ denote the number of label combinations. The local data distribution for the $k$-th vehicle can be characterized by $\mathcal{N}_k = \{ N_k^{1}, N_k^{2}, \ldots, N_k^{M} \}$, where $N_{k}^{m}$ represents the number of samples belonging to label combination $m$ on the $k$-th vehicle. The global distribution is defined as $\mathcal{N}_g = \{ N_g^{1}, \ldots, N_g^{M} \}$, where $N_g^{m} = \sum_{k=1}^{K} N_k^{m}$ denotes the total samples for combination $m$. The label overlap proportion is detailed below.

\textbf{Label Overlap Proportion}: Given $\tilde{N}_i^{m}$ and $\tilde{N}_j^{m}$ denote the normalized occupancy ratios of label combination $m$ on vehicles $i$ and $j$, the label overlap rate is defined as:
\begin{equation}
\zeta = \frac{2}{K(K-1)} \sum_{i=1}^{K-1} \sum_{j=i+1}^{K} \sum_{m=1}^{M} \min\big(\tilde{N}_i^{m}, \tilde{N}_j^{m}\big).
\label{eq21}
\end{equation}
If $\zeta=1$, label distributions of different local datasets are homogeneous. In contrast, they are heterogeneous.

Guided by these metrics, we propose the DA- strategy to mitigate modality imbalance by suppressing dominant modality updates through dynamic modality dropping, and the DA+ strategy to alleviate label imbalance by equalizing the distribution through data generation.

\begin{figure*}[!t]
    \centering
    \includegraphics[width=0.95\textwidth]{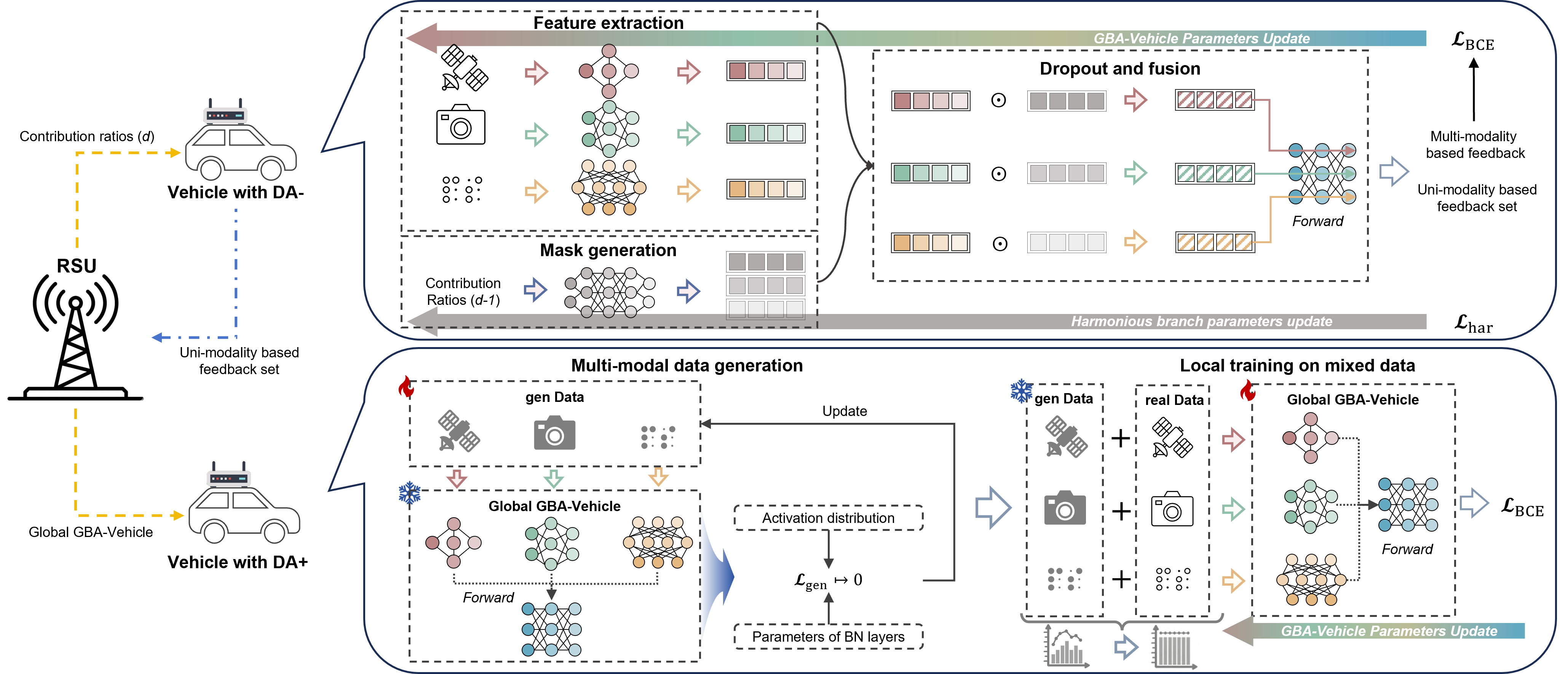}
    \caption{The application illustration of DA- and DA+ strategies.}
    \label{fig3}
    \vspace{-18pt}
\end{figure*}

\vspace{-8pt}
\subsection{DA- for Modality Imbalance}

DA- strategy mitigates the modality imbalance through dynamically modulating update speeds via modality dropping. Let $\mathbf{m}_k^{q,d} \in \mathbb{R}^{L_q}$ denote the dropping mask for modality $q \in \mathcal{Q}_k$ in the $d$-th mini-batch, which takes values of $\mathbf{0}_{L_q}$ or $\mathbf{1}_{L_q}$ to indicate the dropping or keeping of the corresponding feature. As illustrated in Fig. \ref{fig3}, for each vehicle, the mask is generated by the harmonious branch $\mathcal{S}^{\rm H}_k(\cdot)$, a tailored NN mapping the $(d-1)$-th mini-batch modality contribution ratios $\varphi_k^{{q},d-1}$ to the dropping masks. If the vehicle is equipped with all three sensors, the above process is expressed as $\mathbf{M}_k^{\text{H},d} = \mathcal{S}^{\rm H}_k(\boldsymbol{\varphi}_k^{d-1})$ with $\mathbf{M}_k^{d} = \mathbf{m}_k^{\text{R},d} \oplus \mathbf{m}_k^{\text{L},d} \oplus \mathbf{m}_k^{\text{G},d}$ and $\boldsymbol{\varphi}_k^{d-1} = \varphi_k^{{\rm L},d-1} \oplus \varphi_k^{{\rm R},d-1} \oplus \varphi_k^{{\rm G},d-1}$. Within the harmonious branch, contribution ratios are first projected into $L_\text{H}$ dimensions via an MLP and then fed into an attention computation module to integrate cross-modal information. The masked fusion feature is obtained via $\hat{\mathbf{f}}_{k}^{{\rm F},d} = \mathbf{M}_{k}^{d} \odot \mathbf{f}_{k}^{{\rm F},d}$ and utilized to predict the multi-modality based feedback results for BCE loss calculation. Meanwhile, the masked modality features $\hat{\mathbf{f}}_{k}^{q,d} = \mathbf{m}_{k}^{q,d} \odot \mathbf{f}_{k}^{q,d}$ are processed by the corresponding parameter blocks $(\mathbf{W}_{\rm F}^{q}, \mathbf{b}_{\rm F}^{q})$ of the fusion branch to predict uni-modality based feedback vectors, which are sent to the RSU for calculating the ratios of the current mini-batch.

The parameters of each vehicle's harmonious branch update concurrently with the GBA-Vehicle training per mini-batch. Upon getting the contribution ratios of the current mini-batch, the harmonious loss function is calculated as:
\begin{equation}
\label{eq15}
\mathcal{L}_{\mathrm{har},k}
=\sum_{q \in \mathcal{Q}_k}||1-\varphi_k^{q,d}||_2,
\end{equation}
To prevent the harmonious branch from discarding excessive features in pursuit of balance, a regularization function $\mathcal{L}_{\mathrm{har,R}} = -\sum_{q=1}^{|\mathcal{Q}|}R^q$ is introduced. Here, $R^q$ is the sum rates achieved by uni-modality based feedback vectors across all vehicles. Thus, the overall harmonious loss is defined as:
\begin{equation}
\mathcal{L}_{\mathrm{har}}
=\alpha\,\sum_{k=1}^K\mathcal{L}_{\mathrm{har},k}
+ \beta \frac{\mathcal{L}_{\mathrm{har,R}}}{K},
\label{eq17}
\end{equation}
where $\alpha$ and $\beta$ are the trade-off coefficients. With the help of the harmonious branch, the update bias caused by modality imbalance can be alleviated.


\vspace{-8pt}
\subsection{DA+ for Label Imbalance}

To mitigate label imbalance, we propose the DA+ strategy, which leverages privacy-preserving synthetic data generation to drive the label overlap proportion $\zeta$ closer to $1$. Specifically, the running mean and variance of BN layers within GBA-Vehicle, which reflect the input data distribution \cite{CYDG-CVPR-20a}, are used to develop a feature-space distribution matching generation method. Its objective is described as finding synthetic data leading to similar statistics of the BN layers in the global model. As illustrated in Fig. \ref{fig3}, upon receiving the global model, each vehicle extracts the running mean $\{\boldsymbol{\mu}_l\}_{l=1}^{L_{\rm BN}}$ and variance $\{\boldsymbol{\sigma}_l\}_{l=1}^{L_{\rm BN}}$ from the $L_{\rm BN}$ BN layers\footnote{The dimension of the running mean and variance is determined by the number of channels in the layer preceding the BN layer.}. Then, the initialized synthetic dataset of each vehicle $\hat{\mathbb{D}}_k = \{ \mathbf{x}_k^{(q)} \}_{q \in \mathcal{Q}_k}$ is sampled from a Gaussian distribution and fed into the global model with frozen parameters to extract the activations distribution $\{\hat{\boldsymbol{\mu}}_l, \hat{\boldsymbol{\sigma}}_l\}_{l=1}^{L_{\rm BN}}$ for loss function calculation:
\begin{equation}
\begin{aligned}
\mathcal{L}_\mathrm{\text{BN}} = \sum_{l=1}^{{L_{\rm BN}}} \big(|\boldsymbol{\mu}_l 
- \hat{\boldsymbol{\mu}}_l |_2 + |\boldsymbol{\sigma}_l - \hat{\boldsymbol{\sigma}}_l|_2\big)
\label{eq24}
\end{aligned}
\end{equation}
To precisely expand synthetic data with needed label combinations of vehicles, the BCE loss is employed to guide the generation. Finally, the overall generation loss is defined as:
\begin{equation}
\begin{aligned}
\mathcal{L}_\mathrm{\text{gen}} = \mathcal{L}_\mathrm{\text{BN}} + \mathcal{L}_{\mathrm{BCE}}(\mathbf{v}^{\hat{m}}_k, \hat{\mathbf{v}}^{\hat{m}}_k),
\label{eq24-1}
\end{aligned}
\end{equation}
where $\mathbf{v}^{\hat{m}}_k$ is the needed label combination of $k$-vehicle and $\hat{\mathbf{v}}^{\hat{m}}_k$ is the corresponding prediction. After convergence, the local model is trained on the mixed dataset $\hat{\mathbb{D}}_k \cup \mathbb{D}_k$, whose label overlap proportion is closer toward $1$ than the original local dataset, to mitigate the adverse effect of label imbalance.

Next, we analyze the gradient discrepancy between training NN on the synthetic data and the real global data to support the effectiveness of DA+. To facilitate the analysis, we partition the NN into an input block, an output block, and $(L_{\rm BN}-1)$ BN blocks. Let $\hat{\mathbf{G}}_{m,l}$ and $\mathbf{G}_{m,l}^\star$ denote the expected gradients at the $l$-th BN block produced by training on synthetic and real global data with label combination $m$, respectively. Assuming the activation distribution are diagonal Gaussian, the following proposition demonstrates that the gradient error characterizes the convergence bound of generation loss function.

\begin{proposition}
\label{Proposition3}
\emph{[Gradient error bound between synthetic and real data]}: Given the activation distribution $(\hat{\boldsymbol{\mu}}_l, \hat{\boldsymbol{\sigma}}_l)$ before the $l$-th BN block induced by synthetic data, the gradient error between synthetic and real data of this block satisfies:
\begin{equation}
\|\hat{\mathbf{G}}_{m,l} - \mathbf{G}^\star_{m,l}\|_2 \le C_{l,m} \left( \|\boldsymbol{\mu}_l - \hat{\boldsymbol{\mu}}_l\|_2^2 + \|\boldsymbol{\sigma}_l - \hat{\boldsymbol{\sigma}}_l\|_2^2 \right)^{1/2},
\label{eq47}
\end{equation}
where $C_{l,m}$ is a constant. If Eq. \eqref{eq24-1} converges to zero during synthetic data creation, the gradients induced by synthetic data become identical to those induced by real data during training.
\end{proposition}

\emph{Proof}: See Appendix \ref{P3}.

\vspace{-8pt}
\subsection{Coordinated Application Strategy}

To leverage the DA- and DA+ strategies against modality and label imbalances during distributed training of GBA-Vehicle, we tailor a coordinated application strategy for them. In particular, GBA-Vehicle is first trained with the DA- strategy to address modality imbalance, followed by fine-tuning with the DA+ strategy to correct label imbalance.

\textbf{[Step A] GBA-Vehicle training:} The DA- strategy is applied in every mini-batch $d$ of each local training epoch $r$ to modulate GBA-Vehicle updates. During the local forward propagation, each vehicle first inputs the previous contribution ratios to its harmonious branch for dropping masks generation\footnote{The contribution ratios are randomly initialized in the first local epoch. In subsequent epochs, the initial mini-batch inherits ratios of the final mini-batch at the previous local training epoch.}, which are used to drop modality features with the highest contributions, thereby adjusting the gradient allocation dynamically. The masked features are employed to predict the uni-modality based feedback set, which is transmitted to the RSU to compute contribution ratios. Then, the the gradients $g_k^{d,r}$ for each vehicle's harmonious branch are calculated and unicasted back to the respective vehicles for the updates of parameters. In local backward propagation, the harmonious branch is updated via $\mathbf{\Theta}_{s,k}^{{\rm H},d+1,r} = \mathbf{\Theta}_{s,k}^{{\rm H},d,r} -\eta_s\,\nabla_{\mathbf{\Theta}_{s,k}^{{\rm H},d,r}} \mathcal{L}^{d,r}_{\mathrm{har},k}$. Upon completing all mini-batches in the local epoch, accumulated gradients computed through BCE loss are utilized to update parameters of GBA-Vehicle, which are subsequently sent to the RSU for global model aggregation. In contrast, since each harmonious branch is designed for specific modality completeness rate of the corresponding vehicle, it is excluded from global aggregation to preserve the local knowledge. The total process is shown in Algorithm \ref{alg2}.

\textbf{[Step B] GBA-Vehicle fine-tuning:} After mitigating the modality imbalance, the trained GBA-Vehicle is further fine-tuned via the DA+ strategy to address the label imbalance. Each vehicle first evaluates its own imbalance level by computing the label overlap proportion and then identifies the underrepresented label combinations. For each such label combination, the required number of synthetic samples is determined by $\Delta N_k^m = \max(\mathcal{N}_k) - N_k^m$. Considering that the long-tail distribution of $\mathcal{N}_k$ may lead to excessively large $\Delta N_k^m$ values and substantial generation overhead, the RSU helps each vehicle select only a few candidate label combinations for data generation. In particular, the candidate label combinations are ranked according to the corresponding sum rates in the feedback graph, and the top three are retained. Accordingly, the corresponding synthetic data $\hat{\mathbb{D}}_k$ is generated via Eq~\eqref{eq24-1}. Entering the local training stage, the mixed dataset is formed via ${\tilde{\mathbb{D}}}_k=\hat{\mathbb{D}}_k \cup \mathbb{D}_k$ and used to fine-tune GBA-Vehicle. Since the synthetic data primarily compensates for underrepresented label combinations without altering the multi-modal feature space, the previously learned feature extractors remain highly applicable. Consequently, only decision layers mapping feature space to label space are required to update, including the final layer of each extraction branch and the fusion branch. This strategy allows GBA-Vehicle to accommodate the enriched label space while preserving the inter-branch balance established by the DA+ strategy. Finally, the fine-tuned parameters are sent to the RSU for global model aggregation.

\begin{algorithm}[t!]
\caption{DA- Application in Local Training Epoch}
\label{alg2}
\underline{\textbf{Input:}} Local dataset of vehicles $\{\mathbb{D}_k\}_{k=1}^K$, Trained $\mathbf{\Theta}_g$, \\
\phantom{\underline{\textbf{Input:}}} Trained $\{\mathbf{\Theta}^{\text{F},r}_{s,k}, \{\mathbf{\Theta}_s^{q\!,r}\}_{q \in \mathcal{Q}_k}\}_{k=1}^K$ of the $r$-th epoch, \\
\phantom{\underline{\textbf{Input:}}} Trained $\{\mathbf{\Theta}^{\text{H},r}_{k}\}_{k=1}^K$ of the $r$-th epoch.

\underline{\textbf{Output:}} Updated $\{\mathbf{\Theta}^{\text{H},r+1}_k\}_{k=1}^K$, \\
\phantom{\underline{\textbf{Output:}}} Updated $\{\mathbf{\Theta}^{\text{F},r+1}_{s,k}, \{\mathbf{\Theta}_s^{q\!,r+1}\}_{q \in \mathcal{Q}_k}\}_{k=1}^K$.
\begin{algorithmic}[1]
\For{$d = 1,..., N_{\text{mini}}$} \textbf{do}
    \Statex \hspace*{\algorithmicindent} // 1) \textit{Local forward \& multi-modal modulation}
    \For {\textbf{all} vehicle $k$ \textbf{in parallel}}
    \State $\mathbf{M}_k^{\text{H},d,r} = \mathcal{S}^{\text{H},d,r}_k(\boldsymbol{\varphi}_k^{d-1}; \mathbf{\Theta}_{k}^{\text{H},d,r})$;
    \State $\{\mathbf{f}_k^{q,d,r} = \mathcal{S}^{q,d,r}_k(\mathbf{x}^{q}_k);\mathbf{\Theta}_{s,k}^{q, d,r}\}_{q \in \mathcal{Q}_k}$;
    \State $\{\hat{\mathbf{v}}_k^{q,d,r} \gets \mathbf{W}_{\rm F}^{q,d,r}\mathbf{f}_k^{q,d,r}+\mathbf{b}_{\rm F}^{q,d,r}\}_{q \in \mathcal{Q}_k}$;
    \State $\hat{\mathbf{v}}_k^{d,r} \gets \mathcal{S}^{\text{F},d,r}_k(\mathbf{M}_{k}^{d}\odot \mathbf{f}_{k}^{{\text{F}},d,r};\mathbf{\Theta}_{s,k}^{\text{F}, d,r})$;
    \State $\nabla_{\mathbf{\Theta}_{s,k}^{d,r}} \mathcal{L}_{{\rm BCE},k}^{d,r} \gets \frac{\partial \mathcal{L}_{{\rm BCE},k}^{d,r}}{\partial \hat{\mathbf{v}}_k^{d,r}} \frac{\partial \hat{\mathbf{v}}_k^{d,r}}{\partial \mathbf{\Theta}_{s,k}^{d,r}}$;
    \EndFor
\Statex  \hspace*{\algorithmicindent} // 2) \textit{RSU-side contribution computation}
\State $\{\hat{\mathbf{V}}^{q,d,r} \gets \{\hat{\mathbf{v}}_1^{q,d,r}, ..., \hat{\mathbf{v}}_{|\mathcal{K}_{q}|}^{q,d,r}\}\}_{q \in \mathcal{Q}}$;
\State Extract $\{\hat{\mathbf{A}}^{q,d,r}\}_{q \in \mathcal{Q}}$ via Eq.~\eqref{eq6};
\State $\{\hat{\mathbf{T}}^{q,d,r} \gets \mathcal{G}(\hat{\mathbf{A}}^{q,d,r}, \hat{\mathbf{V}}^{q,d,r};\mathbf{\Theta}_g)\}_{q \in \mathcal{Q}}$;
\State Calculate $\{R^{q,d,r}_k\}_{q \in \mathcal{Q}}$ via Eq~\eqref{eq2} for each vehicle;
\State $\varphi^{q,d,r}_k \gets \frac{1}{|\mathcal{Q}_k|-1} \sum_{i=1, j \neq i}^{|\mathcal{Q}_k|}\frac{R^{q,d,r}_k}{R^{q_j,d,r}_k}$;
\State $g_k^{d,r} \gets \frac{\partial \mathcal{L}^{d,r}_{\mathrm{har},k}}{\partial \sum_{q \in \mathcal{Q}_k} \varphi^{q,d,r}_k} \frac{\partial \sum_{q \in \mathcal{Q}_k} \varphi^{q,d,r}_k}{\partial \hat{\mathbf{v}}_k^{q,d,r}}$;
\State \textbf{unicast} $g_k^{d,r}$ and $\{\varphi_k^{q,d,r}\}_{q \in \mathcal{Q}}$ to vehicle $k$
\Statex \hspace*{\algorithmicindent} // 3) \textit{Local backward \& parameters update}
    \For {\textbf{all} vehicle $k$ \textbf{in parallel}}
        \State $\nabla_{\mathbf{\Theta}_{k}^{\text{H},d}} \mathcal{L}^{d,r}_{\mathrm{har},k} \gets g_k \frac{\partial\{\hat{\mathbf{v}}_k^{q,d,r}\}_{q \in \mathcal{Q}_k}}{\partial\mathbf{M}^{\text{H},d,r}_k}\frac{\partial\mathbf{M}^{\text{H},d,r}_k}{\partial \mathbf{\Theta}_{k}^{\text{H},d,r}}$;
        \State $\mathbf{\Theta}_{s,k}^{{\rm H},d+1,r} \gets \mathbf{\Theta}_{s,k}^{{\rm H},d,r} -\eta_s\,\nabla_{\mathbf{\Theta}_{s,k}^{{\rm H},d,r}} \mathcal{L}^{d,r}_{\mathrm{har},k}$;
    \EndFor
\EndFor
\For {\textbf{all} vehicle $k$ \textbf{in parallel}}
\State $\mathbf{\Theta}^{\text{F},r+1}_{s,k} \gets \mathbf{\Theta}^{\text{F},r}_{s,k}-  \eta_s\, \sum_{d=1}^{N_{\text{mini}}} \nabla_{\mathbf{\Theta}_{s,k}^{d,r}} \mathcal{L}_{{\rm BCE},k}^{d,r}$;
\State $\{\mathbf{\Theta}^{{q},r+1}_{s,k} \!\gets\! \mathbf{\Theta}^{{q},r}_{s,k} \!-\! \eta_s\, \sum_{d=1}^{N_{\text{mini}}} \nabla_{\mathbf{\Theta}_{s,k}^{d,r}} \mathcal{L}_{{\rm BCE},k}^{d,r} \}_{q \in \mathcal{Q}_k}$;
\EndFor
\end{algorithmic}
\end{algorithm}

\textit{Remark 2:} The application sequence of DA- and DA+ is immutable. If exchanged, the fine-tuning gradients of the well-converged GBA-Vehicle are too small to support effective modulation. Moreover, the two strategies jointly applied in a single training stage are sub-optimal, since the synthetic data usually provides weaker feature representations than real data. Incorporating it into the DA- training stage would make the dropping behave inconsistently across samples.

\vspace{-8pt}
\section{Simulations}
This section presents the dataset and implementation details of the proposed scheme, followed by a description of the benchmarks and the analysis of the experimental results.

\vspace{-8pt}
\subsection{Dataset}

To showcase the environment perception abilities of distributed multimodal sensing in complex V2X networks, we utilize the real-world dataset FLASH \cite{SGRC-INFOCOM-22}, which was collected by a vehicle equipped with multiple sensors and IEEE 802.11ad Talon Routers operating in the $60$ GHz band over $10$ trials. It covers both LoS and non-line-of-sight (NLoS) scenarios and includes dynamic obstacles such as pedestrians and vehicles. The size of the default codebook is $34$. 

1) \textbf{Label-imbalanced Dataset:}
Due to dynamic obstacles, the label imbalance in the FLASH dataset is so severe that it serves as the basic multimodal data imbalanced scenario. To simulate an RSU serving spatially dispersed vehicles, we introduce a temporal offset for each trial, preserving only the overlapping segments. Finally, the label overlap proportion in the dataset is $\zeta=2.9\%$, indicating severe label imbalance.

2) \textbf{Modality-imbalanced Dataset:}
To further characterize modality imbalance in V2X systems, we consider two representative scenarios including heterogeneous sensor configurations and intermittent modality unavailability. Since GPS is treated as an always-on modality, only RGB and LiDAR data are dropped to achieve completeness ratios of $\kappa^{\rm L}=\kappa^{\rm R}=0.8$.

\begin{figure}[!t]
    \centering
    \includegraphics[width=\columnwidth]{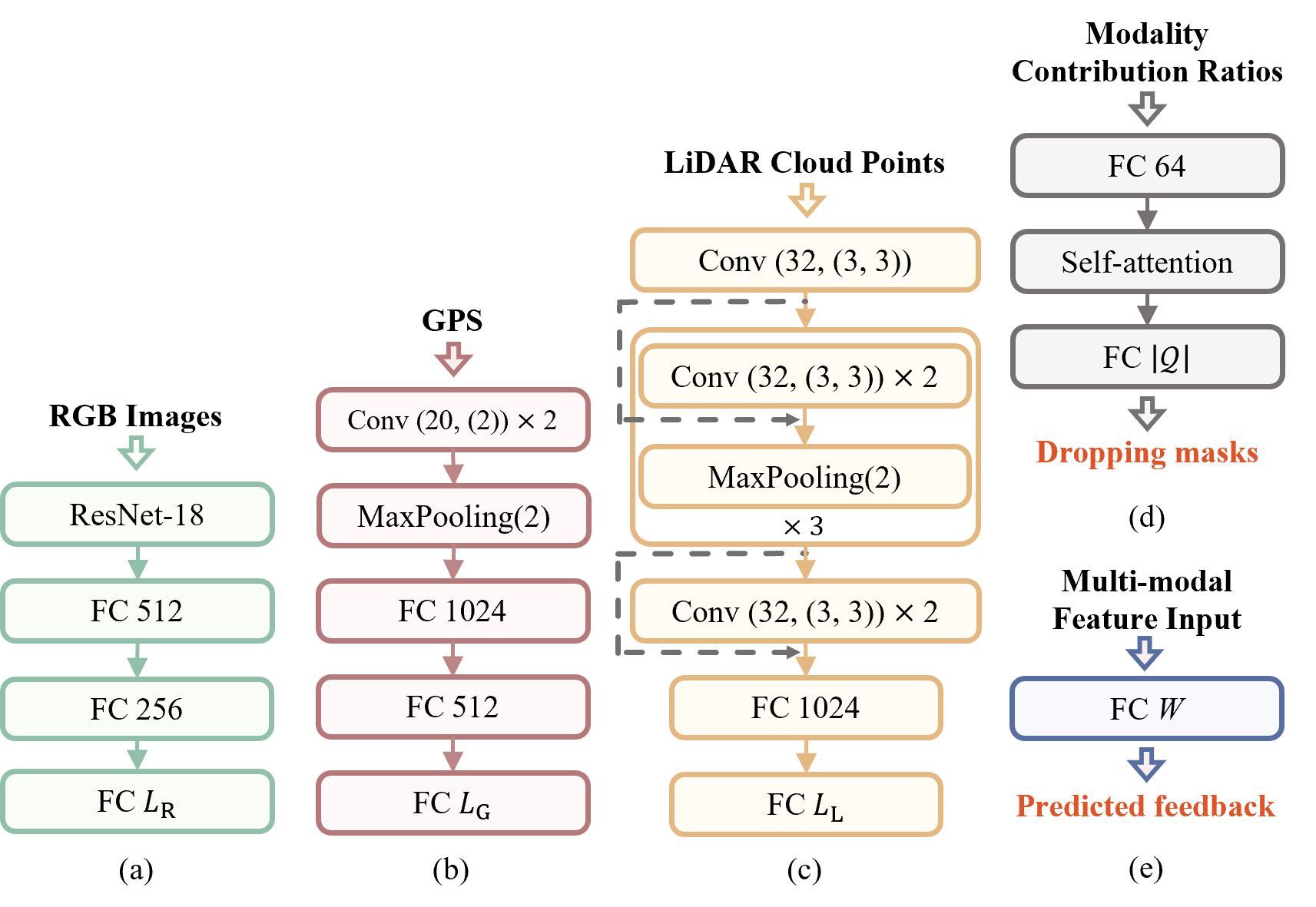}
    \caption{Illustrations of NN architectures for: (a) RGB image branch; (b) GPS branch; (c) LiDAR branch; (d) harmonious branch; (e) multi-modal feature fusion branch.}
    \label{fig4}
\end{figure}

\begin{table}[!t]
	\setlength{\abovecaptionskip}{0.1cm} 
	\renewcommand\arraystretch{0.1} 
    \setlength{\tabcolsep}{4pt}
	\centering
	\caption{Comparisons of delay and feedback overhead among GBA and benchmarks.}
	\resizebox{\linewidth}{!}{
		\label{tab6}
		\begin{tabular}{@{\hspace{1.2em}}c@{\hspace{1.2em}}|@{\hspace{1.2em}}c@{\hspace{1.2em}}|@{\hspace{1.2em}}c@{\hspace{1.2em}}}
			\toprule[0.35mm]
			\textbf{Methods} & \textbf{Delay (ms)} & \textbf{Feedback (bits)} \\
			\midrule[0.15mm]
			\makecell[c]{WMMSE w/ channel estimation} & 20.31 & \makecell[c]{442} \\
			\midrule[0.15mm]
			\makecell[c]{ZF w/ beam sweeping} & 20.31 & \makecell[c]{\underline{\textbf{6}}} \\
			\midrule[0.15mm]
			\makecell[c]{GBA} & \underline{\textbf{0.91}} & \makecell[c]{34} \\
			\bottomrule[0.35mm]
		\end{tabular}
	}
\end{table}

\begin{figure}[!t]
\centering
\includegraphics[width=0.85\columnwidth]{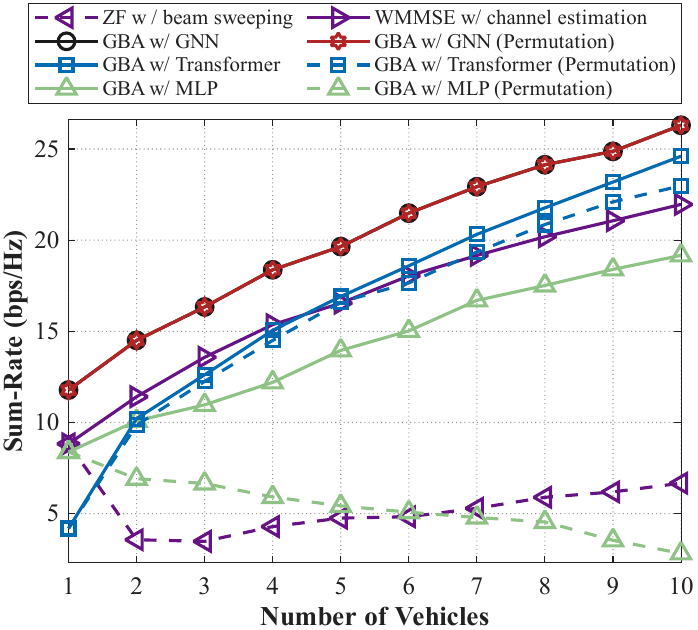}
\caption{Effective sum rate comparison among various GBA-RSU variants for different numbers of vehicles.}
\label{fig6}
\vspace{2pt}
\end{figure}

\vspace{-8pt}
\subsection{Implementation Details}

Both GBA-RSU and GBA-Vehicle utilize AdamW optimizer with initial learning rates of $\eta_g=\eta_{g-r}=1\times 10^{-3}$ and $\eta_s=1\times 10^{-4}$. We configure the GBA-RSU hidden dimension as $d_{\text{g}} = 384$ and the GBA-Vehicle feature dimensions as $L_\text{G} = 256$, $L_\text{R} = 256$, $L_\text{L} = 512$ and $L_\text{H} = 64$. The corresponding batch sizes are $2048$ and $256$. Fig. \ref{fig4} depicts the architectures for the extraction and fusion branches of GBA-Vehicle, and the harmonious branch. For GBA-RSU training, the dropping and error probabilities are set to $p_{\text{drop}}=0.25$ and $p_{\text{error}}=0.2$. For GBA-Vehicle, the trade-off coefficients are set $\alpha=\beta=1$ and the synthetic data is generated by AdamW optimizer with a batch size of $32$.

The number of epochs for GBA-RSU training and retraining is set to $30000$. For DA- strategy-based training and DA+ strategy-based fine-tuning of GBA-Vehicle, the required global epochs vary according to the sensor configuration. Specifically, we allocate $200$ global epochs for the label-imbalanced dataset, while assigning $120$ and $250$ epochs for the sensor type and performance imbalanced scenarios, respectively. Local training epochs for both stages are set to $5$ per vehicle. Each vehicle's local data is split into an $80\%$ training set and a $20\%$ test set.

\vspace{-8pt}
\subsection{Benchmarks}

We compare the GBA scheme against two representative alignment benchmarks based on high-resolution Type-II codebook-based CSI feedback \cite{LLWH-ComSM-24}. The DA strategy is evaluated under three distributed training paradigms.

\begin{itemize}
 
    \item \textbf{WMMSE w/ channel estimation \cite{PeJa-TSP-22}:}
    The WMMSE method with CSI estimated by RSS measurements via greedy algorithm is used as a benchmark scheme. It is adopted as the upper bound of the sum rates of traditional methods as it provides a very decent local optimum.

    \item \textbf{ZF w/ beam weeping \cite{JoUN-TSP-05}:}
    The ZF method with beam sweeping is used as a benchmark scheme. Each vehicle selects the beam index with the highest RSS and sends it to the RSU, which then allocates power equally.

    \item \textbf{Upper bound of FL (UB):}
    All vehicles transmit their raw local data to the RSU for GBA-Vehicle training. This approach serves as the upper bound for FL, but at the cost of prohibitive data transmission overhead.

    \item \textbf{FedAvg \cite{MMRH-arXiv-16}:}
    Each vehicle trains local GBA-Vehicle on private dataset and uploads the parameters to the RSU for global aggregation. It serves as a benchmark for FL.

    \item \textbf{Federated Learning for Automated Selection of High-band mmWave Sectors (FLASH) \cite{SGRC-INFOCOM-22}:}
    The FLASH is an FL-based alignment benchmark scheme, where RSU selects updated branch for each vehicle.

\end{itemize}

\begin{table*}[!t]
    \setlength{\abovecaptionskip}{0.04cm}
    \renewcommand\arraystretch{1.1}
    \setlength{\tabcolsep}{3.5pt}
    \setlength{\aboverulesep}{2.8pt}
    \setlength{\belowrulesep}{1.4pt}
    \caption{Comparisons among DA and benchmarks under different imbalanced scenarios.}
    \centering
    \fontsize{4.5pt}{5.2pt}\selectfont
    \label{tab2}
    \resizebox{0.91\linewidth}{!}{
    \begin{tabular}{c|C{1.0cm}|C{1.0cm}|C{1.0cm}|C{1.0cm}|C{1.0cm}|C{1.0cm}}
        \toprule[0.2mm]
        \multicolumn{1}{c|}{} & \multicolumn{6}{c}{\textbf{Label Imbalanced Offline Training Dataset}} \\[-2pt]
        \cmidrule[0.1mm]{2-7}
        \multicolumn{1}{c|}{\textbf{Schemes}} &
          \multicolumn{2}{c|}{\textbf{~~~~~~~~Modality balance~~~~~~~~}} &
          \multicolumn{2}{c|}{\textbf{~~~~~Sensor type imbalance~~~~~}} &
          \multicolumn{2}{c}{\textbf{Sensor performance imbalance}} \\[-2.2pt]
        \cmidrule[0.1mm]{2-3}\cmidrule[0.1mm]{4-5}\cmidrule[0.1mm]{6-7}
        & \textbf{sum rates $\uparrow$} & \textbf{Accuracy $\uparrow$}
        & \textbf{sum rates $\uparrow$} & \textbf{Accuracy $\uparrow$}
        & \textbf{sum rates $\uparrow$} & \textbf{Accuracy $\uparrow$} \\[-2.2pt]
        \midrule[0.1mm]
        \makecell[c]{GBA w/ UB} &
        29.65 & 62.06 &
        24.47 & 43.10 &
        24.39 & 35.11 \\[-3pt]
        \midrule[0.1mm]
        \textbf{GBA w/ DA} &
        \underline{\textbf{26.30}} & \underline{\textbf{48.97}} &
        \underline{\textbf{22.37}} & \underline{\textbf{33.47}} &
        \underline{\textbf{23.48}} & \underline{\textbf{33.15}} \\[-3pt]
        \midrule[0.1mm]
        \makecell[c]{GBA w/ FedAvg} &
        24.20 & 38.94 &
        19.05 & 22.96 &
        18.90 & 20.03 \\[-3pt]
        \midrule[0.1mm]
        \makecell[c]{GBA w/ FLASH} &
        23.07 & 31.30 &
        17.74 & 15.85 &
        18.23 & 16.90 \\[-3pt]
        \bottomrule[0.2mm]
    \end{tabular}
}
\vspace{-7pt}
\end{table*}

\begin{figure*}[!t]
    \centering
    \includegraphics[width=0.91\textwidth]{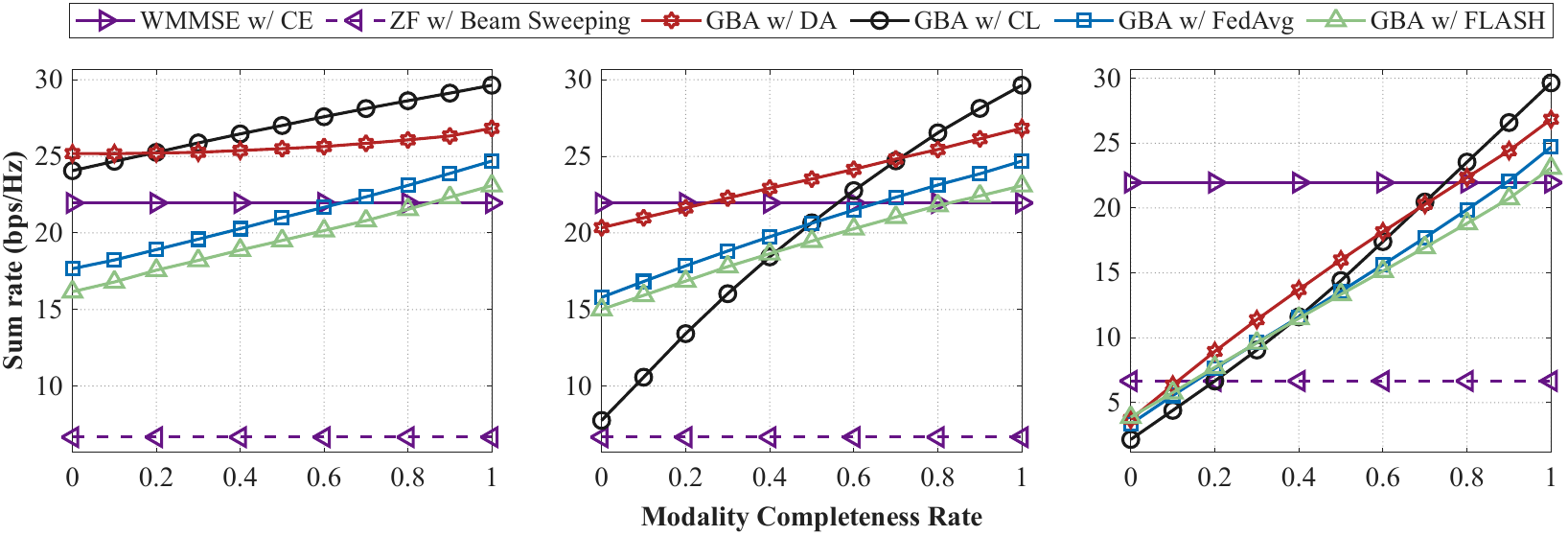}
    \caption{Comparisons among GBA and benchmarks under varying online sensor configurations: modality completeness rates decreasing for RGB (Left), LiDAR (Middle), and both of them (Right).}
    \label{fig5}
    \vspace{-16pt}
\end{figure*}

\vspace{-8pt}
\subsection{Performance Metric}

We use prediction accuracy to evaluate the performance of feedback generation, while the effective sum rates is employed to assess the communication performance of beam alignment. Let $T_{\rm coher} = \frac{T_{\rm contact}}{W}$ denote the beam coherence time of the system, namely, the duration for which a beam remains valid \cite{Sale-TON-24}. Here, $T_{\rm contact}$ represents the total time during which the vehicle remains within the coverage of the RSU, defined as:
\begin{equation}
\label{eq44}
    T_{\rm contact} = \frac{2h \tan(\frac{\phi_{\rm RSU}}{2})}{v},
\end{equation}
where $\rm RSU$ and $h$ denote the coverage angle and height of the RSU, respectively, and $v$ is the velocity of the vehicle \cite{MMGE-INFOCOMWkshps-19}. Then, the beam alignment delay must satisfy $T_{\rm delay} < T_{\rm coher}$, and the effective sum rates is defined as:
\begin{equation}
\label{eq32}
    \tilde{R} = \frac{T_{\rm coher}-T_{\rm delay}}{T_{\rm coher}} R
\end{equation}
Here, $T_{\rm delay}$ and $R$ denote the length of alignment period and the theoretical sum rates, respectively. 

\vspace{-8pt}
\subsection{Comparative Study}

We begin by comparing the delay and feedback overhead. According to \cite{Sale-TON-24}, the beam coherence time and conventional beam initialization time in the FLASH dataset are $62.4$ ms and $20.31$ ms, respectively. In the proposed GBA scheme, beam initialization is achieved through feedback prediction, which requires only about $0.91$ ms. This reduces the initialization delay by more than $95\%$ and leaves more time for reliable transmission within the beam coherence interval. For feedback overhead, to enable high-quality multi-user alignment, the CE method requires the full RSS vector, whose average quantized length is $442$ bits in the FLASH dataset. In contrast, the beam sweeping scheme only needs to transmit $6$ bits for the selected beam index, but it is less effective in mitigating severe inter-vehicle interference. Although the proposed GBA scheme requires $34$ bits to transmit the predicted feedback, which is higher than that of beam sweeping, it provides richer feedback information for interference mitigation, thereby improving the system sum rates. Moreover, as the codebook size increases, the overhead of full feedback grows substantially, making the overhead advantage of GBA more significant. Since the FLASH dataset is based on the $5$ GHz back channel of the Talon routers \cite{SGRC-INFOCOM-22}, The feedback latencies for CE, beam sweeping, and GBA are calculated as $247~ns$, $4~ns$ and $19~ns$. When the entire alignment procedure is taken into account, the corresponding alignment periods occupy $32\%$, $32\%$, and $2\%$ of the coherence time, respectively. 

Fig.~\ref{fig6} presents the communication and scalability performance. Due to reduced alignment overhead, the GBA scheme achieves a higher effective sum rates than WMMSE and ZF across different vehicle counts. To further examine the effect of the RSU-side policy architecture, we replace the GNN-based GBA-RSU with Transformer-based and MLP-based variants, while training GBA-Vehicle using the DA strategy. Although the Transformer can handle variable-sized inputs and therefore exhibits limited scalability, its performance deteriorates under vehicle reordering because it doesn't satisfy permutation equivariance. Likewise, even when separate MLPs are trained for different feedback graph sizes, the MLP-based GBA remains highly sensitive to input permutations. This issue becomes more severe as the number of vehicles increases, since larger feedback graphs amplify the distribution skew caused by different input orderings. In contrast, the GNN-based GBA is insensitive to vehicle reordering while maintaining the best scalability. These results demonstrate that the proposed GBA scheme provides a scalable and low-overhead solution for beam alignment in V2X systems.

\begin{table*}[!t]
		\setlength{\abovecaptionskip}{0.04cm} 
		\renewcommand\arraystretch{1}
        \setlength{\tabcolsep}{4pt}
		\setlength{\aboverulesep}{2.8pt}
		\setlength{\belowrulesep}{1.4pt}
		\tiny
		\caption{Ablation study of GBA-RSU structures. Performance is reported as the sum rates (bps/Hz).}
		\centering
		\label{tab3}
		\resizebox{\linewidth}{!}{
			\begin{tabular}{c|c|c|c|c|c|c|c|c|c|c}
				\toprule[0.2mm]
				\multirow{2}{*}{\textbf{Schemes}}  & \multicolumn{10}{c}{\textbf{Number of vehicles}} \\[-2.6pt]  
				\cmidrule[0.1mm]{2-11}
				& 1 & 2 & 3 & 4 & 5 & 6 & 7 & 8 & 9 & 10\\ [-2.2pt]  
				\midrule[0.1mm]   
				{w/ non-negativity and re-normalization} & \textbf{\underline{12.79}} & \textbf{\underline{14.76}} & \textbf{\underline{15.38}} & 
                \textbf{\underline{22.26}} & 
                \textbf{\underline{23.15}} & 
                \textbf{\underline{24.27}} & 
                \textbf{\underline{26.41}} & 
                \textbf{\underline{27.69}} & \textbf{\underline{28.17}} & \textbf{\underline{29.12}} \\[-2.8pt]
				\midrule[0.1mm]
				\makecell[c]{w/ re-normalization} & 12.79 & 14.31 & 14.57 & 21.94 & 23.13 & 23.97 & 25.44 & 27.14 & 27.79 & 28.87 \\[-3pt]
				\midrule[0.1mm]
				\makecell[c]{w/ re-non-negativity} & 12.73 & 13.71 & 14.36 & 17.29 & 20.91 & 20.68 & 22.94 & 25.15 & 25.96 & 28.27 \\[-3pt]
				\midrule[0.1mm]
				\makecell[c]{w/o non-negativity and re-normalization} & 12.73 & 8.78 & 10.31 & 13.31 & 13.60 & 12.54 & 13.36 & 21.44 & 14.90 & 13.22 \\[-3pt]
				\midrule[0.1mm]
				\makecell[c]{WMMSE w/ channel estimation} & 13.09 & 16.92 & 20.13 & 22.78 & 24.53 & 26.75 & 28.41 & 29.93 & 31.23 & 32.56 \\[-3pt]
				\bottomrule[0.2mm]
			\end{tabular}
		}
	\end{table*}

\begin{table*}[!t]
		\setlength{\abovecaptionskip}{0.04cm} 
		\renewcommand\arraystretch{1}
        \setlength{\tabcolsep}{6pt}
		\setlength{\aboverulesep}{2.8pt}
		\setlength{\belowrulesep}{1.4pt}
		\tiny
		\caption{Ablation study of GBA-RSU offline training schemes, with performance measured by sum rates (bps/Hz).}
		\centering
		\label{tab4}
		\resizebox{\linewidth}{!}{
			\begin{tabular}{c|c|c|c|c|c|c|c|c|c|c}
				\toprule[0.2mm]
				\multirow{2}{*}{\textbf{Schemes}}  & \multicolumn{10}{c}{\textbf{Number of vehicles}} \\[-2.6pt]  
				\cmidrule[0.1mm]{2-11}
				& 1 & 2 & 3 & 4 & 5 & 6 & 7 & 8 & 9 & 10\\ [-2.2pt]  
				\midrule[0.1mm]   
				{w random dropping} & \textbf{\underline{12.81}} & \textbf{\underline{17.79}} & \textbf{\underline{20.03}} & \textbf{\underline{23.44}} & \textbf{\underline{25.32}} & \textbf{\underline{27.57}} & \textbf{\underline{30.64}} & \textbf{\underline{32.53}} & \textbf{\underline{33.84}} & \textbf{\underline{35.92}} \\[-2.8pt]
				\midrule[0.1mm] 
				\makecell[c]{w/o random dropping} & 12.79 & 14.76 & 15.38 & 21.18 & 23.15 & 23.77 & 25.28 & 27.69 & 28.17 & 29.12 \\[-3pt]
				\midrule[0.1mm]
				\makecell[c]{WMMSE w/ channel estimation} & 13.09 & 16.92 & 20.13 & 22.78 & 24.53 & 26.75 & 28.41 & 29.93 & 31.23 & 32.56 \\[-3pt]
				\bottomrule[0.2mm]
			\end{tabular}
		}
		\vspace{-10pt}
	\end{table*}

\begin{table}[!t]
		\setlength{\abovecaptionskip}{0.1cm} 
		\renewcommand\arraystretch{0.1} 
        \setlength{\tabcolsep}{4pt}
		\centering
		\caption{Comparisons under three sensor configurations among different deployment strategies.}
		\resizebox{\linewidth}{!}{
			\label{tab5}
			\begin{tabular}{@{\hspace{1.2em}}c@{\hspace{1.2em}}|@{\hspace{1.2em}}c@{\hspace{1.2em}}}
				\toprule[0.35mm]
				{\textbf{Deployment Scheme}}  &\makecell[c]{\textbf{sum rates performance} (bps/Hz)\\(Sensor configuration: L-G~/~R-G~/~G)}  \\
				\midrule[0.15mm]
				\makecell[c]{\textbf{Pruning-based}} & \makecell[c] {\underline{\textbf{25.19}}~/~\underline{\textbf{20.39}}~/~\underline{\textbf{3.71}}} \\ 
				\midrule[0.15mm]
				\makecell[c]{Sensor-based} & \makecell[c]{19.77~/~~7.97~~/~2.57 } \\ 
				\midrule[0.15mm]	
				\makecell[c]{Model-based} 	& 16.22~/~11.33~/~1.84  \\
				\bottomrule[0.35mm]
			\end{tabular}
		}
        \vspace{2pt}
\end{table}

Then, the effectiveness of DA strategy in mitigating label and modality imbalances is analyzed. Table \ref{tab2} compares DA with benchmark distributed training strategies under three extremely imbalanced scenarios. Compared with FLASH, DA improves the feedback accuracy of GBA-Vehicle by $56.45\%$ under label imbalance only. When modality imbalnce additionally presenting, the improvement further rises to nearly $100\%$. These results indicate that DA strategy improves the coverage of the GBA-Vehicle label space and alleviates performance degradation caused by biased aggregation under label imbalance. Furthermore, to examine whether DA also promotes balanced updates across modality-specific extraction branches, we evaluate GBA-Vehicle trained with different benchmarks under varying degrees of modality missing, which reflect online heterogeneous vehicle configurations in V2X. To reduce the randomness of missing-data selection, we conduct repeated tests under both sensor heterogeneity and partial data loss. Specifically, for sensor heterogeneity, we enumerate all vehicle combinations calculated by $\sum_{q \in \mathcal{Q}} C_{10}^{\kappa^{q} \times 10}$, and for partial data loss, we conduct the same number of Monte Carlo trials. 

As shown in Fig.~\ref{fig5}, when the RGB modality completeness rate decreases to $0\%$, the effective sum rates of GBA-Vehicle trained with CL and FLASH drop by more than $18\%$ and $28\%$, respectively. In contrast, GBA-Vehicle trained with DA shows stronger resilience, with reductions of only $6\%$ and $9\%$ under the corresponding two settings. Similarly, when the LiDAR completeness ratio decreases to $0\%$, the performance of GBA-Vehicle with CL drops by nearly $74\%$, indicating that it relies heavily on LiDAR structural information and is therefore highly sensitive to LiDAR absence. By contrast, for GBA-Vehicle trained with DA, the degradation is limited to $21\%$ because DA better exploits RGB scattering distribution information. Notably, when the LiDAR completeness ratio falls below $0.7$, the achieved sum rates of GBA-Vehicle trained with DA even surpasses that achieved by GBA-Vehicle with CL. Additionally, simultaneous decreases in RGB and LiDAR completeness substantially degrade all methods. This is because relying only on GPS-based relative positioning fails to capture NLoS structures in complex urban environments, resulting in higher feedback prediction error and a lower system sum rates. Nevertheless, GBA-Vehicle trained with DA remains more robust than the other benchmarks. These results show that the proposed DA strategy effectively addresses label and modality heterogeneity in distributed multimodal sensing.

Table \ref{tab3} presents the ablation results of GBA-RSU, evaluating the effects of the non-negativity and re-normalization processes in the beam projection layer. GBA-RSU is trained on ten datasets with varying feedback-graph sizes to prevent overfitting and is evaluated in terms of sum rates with an approximately $0$ ms alignment period, so that the impact of GBA-Vehicle feedback prediction is excluded. The results indicate that omitting both processes during training violates the constraints in Eq. \eqref{eq3d} and causes model failure. While incorporating either process independently improves performance, re-normalization yields greater performance gains than non-negativity as vehicle density increases. This is because a larger number of vehicles leads to more severe inter-user interference. The re-normalization process effectively mitigates this effect through vertex pruning. Additionally, Table \ref{tab4} demonstrates that the stochastic topology perturbation strategy effectively simulates variations in the number of vehicles, thereby improving the scalability of GBA-RSU and enabling it to outperform the WMMSE w/ CE.

Table \ref{tab5} presents the performance comparison of different deployment strategies across vehicular sensor configurations during online deployment. We evaluate two common baseline strategies. The model-based strategy transmits the entire trained GBA-Vehicle to newly arriving vehicles by zero-padding the channels of missing modalities to maintain input consistency. The sensor-based strategy transmits the fusion and extraction branches associated with the available sensors with zero-padding applied to the channels of missing features However, inherent non-zero bias parameters in extraction and fusion branch transform these zero inputs into noise, degrading feedback prediction accuracy and reducing the effective sum rates. Our pruning-based strategy explicitly removes fusion network weights associated with missing modalities. This ensures fusion occurs strictly among available modalities, preventing noise introduction. The results illustrate that our strategy improves the sum rates by $78\%$ over the model-based strategy across all sensor configurations.

\vspace{-4pt}
\section{Conclusions}
This paper presented a multimodal distributed GBA scheme for scalable and low-overhead beam management in V2X systems. The architecture combined a GBA-Vehicle unit for local feedback prediction by multimodal sensing, along with a GBA-RSU unit for dynamic alignment. A dedicated DA strategy mitigated the multimodal data imbalance by addressing modality imbalance via dominant modality dropout and label imbalance via underrepresented data generation. Tailored offline training and online execution procedures coordinated these components. Extensive simulations on real-world driving datasets results showed that GBA reduces alignment overhead by $93.75\%$ versus high-resolution type II codebook feedback and achieved higher effective sum rates than WMMSE with CE. Under extreme label imbalance, the DA strategy improved feedback prediction accuracy of GBA-Vehicle by $56.45\%$ over FLASH, surging to nearly $100\%$ when modality imbalance was also introduced. These results confirmed that the proposed scheme provides a robust solution for beam management in complex vehicular environments.

\vspace{-8pt}
\bibliographystyle{IEEEtran}
\bibliography{bibtex/bib/IEEEabrv,bibtex/bib/IEEEexample}

\vspace{-8pt}
\appendix
\subsection{Proof of Proposition 1}
\label{P1}

\emph{Invariance of the constraints.}
For the power constraint Eq. \eqref{eq3d}, under the transformed allocation matrix $\Pi^\top \mathbf P \Pi$, we have $\mathrm{Tr}(\Pi^\top \mathbf P \Pi)=\mathrm{Tr}(\mathbf P)\le P_{\max}$, where the equality follows from the invariance of the trace under similarity transformations. For the beam selection constraint Eq. \eqref{eq3e}, under the transformed selection matrix $\mathbf U \Pi$, we obtain $(\mathbf U \Pi)^\top \mathbf 1_W=\Pi^\top \mathbf U^\top \mathbf 1_W=\Pi^\top \mathbf 1_K=\mathbf 1_K$, 
because multiplying the all-ones vector by a permutation matrix only reorders its entries and therefore leaves it unchanged. Thus, Eq. \eqref{eq3e} is preserved.

\emph{Invariance of the sum-rate objective.}
The objective in Eq. \eqref{eq3} is a summation over all user nodes, denoted as $\sum_{k=1}^K R_k$. The achievable rate of user $k$ is given by Eq. \eqref{eq2}. Let $\pi(\cdot)$ denote the permutation of the index set $\{1,\ldots,K\}$ induced by the permutation matrix $\Pi$. The graph after permutation is understood as a relabeled version of the original graph. Accordingly, throughout the proof, any quantity indexed by permuted node labels refers to the corresponding quantity in the permuted graph rather than the quantity at the same index in the original graph. Under the permutation $\Pi$, all user-related quantities are consistently relabeled according to the permutation. Since $\pi(\cdot)$ is a bijection on the index set, the permutation only reorders the summands in the finite sum over user nodes. Therefore, the sum rates objective is expressed by
\begin{equation}
\sum_{k=1}^K R_{\pi(k)}
= \sum_{k=1}^K R_k.
\end{equation}

\vspace{-8pt}
\subsection{Proof of Proposition 2}
\label{P2}
Since GBA-RSU comprises cascaded stages, the entire network should satisfy permutation equivariance if each constituent stage is equivariant.
 
\emph{Graph construction.}
Given the permuted inputs $\mathbf{V}\Pi$, transformed adjacency matrix is obtained as $\Pi^\top\mathbf{A}\Pi$ via Eq. \eqref{eq6}, confirming the permutation equivariance of the graph construction process.

\emph{Edge encoding and vertex encoding.}
Given the permuted inputs $\mathbf{V}\Pi$, the original edge features are expressed as:
\begin{equation}
    \mathbf{v}_{\pi(i),\pi(j)} = \mathbf{v}_{\pi(i)} \oplus \mathbf{v}_{\pi(j)}=\mathbf{v}_{i,j}, \quad (i,j) \in {\mathcal E}.
\end{equation}
Since the Edge MLP $\mathcal{G}^{\text{E}}(\cdot)$ is shared among edges, the corresponding extracted interference information of $k$-th vehicle satisfy $\mathbf f_{\pi(k),\pi(i)}^{\rm E}=\mathbf f_{k,i}^{\rm E}$. Based on the permutation equivariant of adjacency matrix, we can get the size of egde set is unchanged after transforming. Accordingly, the averagely aggregated interference information satisfies
\begin{equation}
\overline{\mathbf f}_{\pi(k)}^{\rm E}
=\frac{\sum_{({\pi(k)},i) \in \mathcal{E}_{\pi(k)}} {\mathbf f_{{\pi(k)},i}^{\rm E}}}
     {|\mathcal E_{\pi(k)}|}
=\frac{\sum_{(k,j)\in {\mathcal E}_{k}}{\mathbf f}_{k,j}^{\rm E}}{|{\mathcal E}_{k}|}
=\overline{\mathbf f}_{k}^{\rm E},
\end{equation}
Hence, the summation over neighboring nodes is unchanged except for the order of the summands. Subsequently, the following node-level update relays on shared the self-vertex MLP $\mathcal{G}^{\rm SV}(\cdot)$ and cross-vertex MLP $\mathcal{G}^{\rm CV}(\cdot)$. Consequently, the permutation equivariance of both edge and vertex encoding layers is established.

\emph{Beam projection.}
Similarly, the node-level operation at this layer would not change the equivariance, thereby indicating that the transformed masked vector satisfies $\hat{\mathbf z}_{\pi(k)} = \hat{\mathbf z}_k$. Accordingly, we can get $\sum_{k=1}^{K}\|\hat{\mathbf z}_{\pi(k)}\|_2^2 =\sum_{j=1}^{K}\|\hat{\mathbf z}_k\|_2^2$, which suppors the re-normalization operation. Currently, we can get:
\begin{equation}
    {\hat{\mathbf t}}_{\pi(k)}
    =\sqrt{\frac{P_{\text{max}}}{\sum_{k=1}^{K}\|\hat{\mathbf z}_{\pi(k)}\|_2^2}} \hat{\mathbf z}_{\pi(k)}
    = \sqrt{\frac{P_{\text{max}}}{\sum_{k=1}^{K}\|\hat{\mathbf z}_k\|_2^2}} \hat{\mathbf z}_k
    ={\hat{\mathbf t}}_{k}.
\end{equation}
Stacking these columns yields $\tilde{\hat{\mathbf T}}=\hat{\mathbf T}\Pi$, thereby verifying the permutation equivariance of GBA-RSU.

\vspace{-8pt}
\subsection{Proof of Proposition 3}
\label{P3}
We begin by introducing the forward activations and backpropagation gradients of NN. For a sample with the $m$-th label, let $\mathbf{a}_{m,l} = \Phi_l(\mathbf{x}) \in \mathbb{R}^C$ denote the input activation before the $l$-th BN block which have $C$ channels, where $\Phi_l(\cdot)$ represents the feature mapping from the input space to the activation space of that block. Let $\mathbf{g}_{m,l}$ denote the corresponding backpropagated gradient $\nabla_{\mathbf{\Theta}_l} \mathcal{L}_{\rm BCE}(\mathbf{\Theta}_l; \mathbf{a}_{m,l})$. Then, the gradient can be viewed as a function of the activation, written as $\mathbf{g}_{m,l} = \psi_{m,l}(\mathbf{a}_{m,l}) \in \mathbb{R}^C$, where $\psi_{m,l}(\cdot)$ is the activation-to-gradient mapping. We assume that $\psi_{m,l}(\cdot)$ is $C_{l,m}$-Lipschitz.

To measure the discrepancy between the expected gradients induced by synthetic and real data, we first project the vector-valued gradient mapping onto an arbitrary unit direction. Specifically, for any unit vector $\mathbf{e} \in \mathbb{R}^C$, define the scalar function $f_\mathbf{e}(\mathbf{a}_{m,l}) = \mathbf{e}^\top \psi_{m,l}(\mathbf{a}_{m,l})$. Then, for any two activation $\mathbf{a}_{m,l}, \mathbf{a}'_{m,l}$, we have:
\begin{equation}
    \begin{aligned}
|f_\mathbf{e}(\mathbf{a}_{m,l}) - f_\mathbf{e}(\mathbf{a}'_{m,l})| &= |\mathbf{e}^\top (\psi_{m,l}(\mathbf{a}_{m,l}) - \psi_{m,l}(\mathbf{a}'_{m,l}))| \\
&\le \|\mathbf{e}\|_2 \cdot \|\psi_{m,l}(\mathbf{a}_{m,l}) - \psi_{m,l}(\mathbf{a}'_{m,l})\|_2 \\
&\le C_{l,m} \|\mathbf{a}_{m,l} - \mathbf{a}'_{m,l}\|_2,
\end{aligned}
\end{equation}
where the second step follows from the Cauchy-Schwarz inequality and the last step follows from the Lipschitz assumption on $\psi_{m,l}(\cdot)$.

Now, let $\mathcal{D}^\star_{l,m}$ and $\hat{\mathcal{D}}_{l,m}$ denote the activation distributions before the $l$-th BN block induced by real data and synthetic data, respectively. The Kantorovich-Rubinstein inequality \cite{BhJL-ExpoMath-19} implies that the difference between its expectations under these two distributions is bounded by their $W_1$ distance, expressed as $\left| \mathbb{E}_{\hat{\mathbf{a}}_{m,l} \sim \hat{\mathcal{D}}_{l,m}}[f_\mathbf{e}(\hat{\mathbf{a}}_{m,l})] - \mathbb{E}_{\mathbf{a}^\star_{m,l}\sim\mathcal{D}^\star_{l,m}}[f_\mathbf{e}(\mathbf{a}^\star_{m,l})] \right| \le C_{l,m} W_1(\hat{\mathcal{D}}_{l,m}, \mathcal{D}^\star_{l,m})$. Next, defining the expected gradients induced by real and synthetic data as $\mathbf{G}^\star_{m,l}$ and $\hat{\mathbf{G}}_{m,l}$, the above inequality can be rewritten as $\left| \mathbf{e}^\top (\hat{\mathbf{G}}_{m,l} - \mathbf{G}^\star_{m,l}) \right| \le C_{l,m} W_1(\hat{\mathcal{D}}_{l,m}, \mathcal{D}^\star_{l,m})$. Since this holds for any unit vector, taking the supremum over all unit vectors and using the dual representation of the Euclidean norm yields
\begin{equation}
    \begin{aligned}
\|\hat{\mathbf{G}}_{m,l} - \mathbf{G}^\star_{m,l}\|_2 
&= \sup_{\|\mathbf{e}\|_2=1} |\mathbf{e}^\top \hat{\mathbf{G}}_{m,l} - \mathbf{G}^\star_{m,l}| \\
&\le C_{l,m} W_1(\hat{\mathcal{D}}_{l,m}, \mathcal{D}^\star_{l,m}) \\
&\le C_{l,m} W_2(\hat{\mathcal{D}}_{l,m}, \mathcal{D}^\star_{l,m}),
\end{aligned}
\end{equation}
where the last step follows from \cite{BhJL-ExpoMath-19}. Under the diagonal Gaussian assumption, the activation distributions are expressed as $\mathcal{D}^\star_{l,m} = \mathcal{N}(\boldsymbol{\mu}_l, \Sigma_l)$ and $\hat{\mathcal{D}}_{l,m} = \mathcal{N}(\hat{\boldsymbol{\mu}}_l, \hat{\Sigma}_l)$. The corresponding covariance matrices are defined as $\Sigma_l = \text{diag}(\boldsymbol{\sigma}_l^2)$ and $\hat{\Sigma}_l = \text{diag}(\hat{\boldsymbol{\sigma}}_l^2)$. Accordingly, the matrix square roots simplify to $\Sigma_l^{1/2} = \text{diag}(\boldsymbol{\sigma}_l)$ and $(\Sigma_l^{1/2} \hat{\Sigma}_l \Sigma_l^{1/2})^{1/2} = \text{diag}(\boldsymbol{\sigma}_l \odot \hat{\boldsymbol{\sigma}}_l)$. Applying the Bures-Wasserstein \cite{PaZe-ARSA-19} to these diagonal distributions yields the explicit 2-Wasserstein distance:
\begin{equation}
\begin{aligned}
W_2^2(\hat{\mathcal{D}}_{l,m}, \mathcal{D}^\star_{l,m}) = \|\boldsymbol{\mu}_l - \hat{\boldsymbol{\mu}}_l\|_2^2 + \|\boldsymbol{\sigma}_l - \hat{\boldsymbol{\sigma}}_l\|_2^2.
\end{aligned}
\end{equation}
Substituting this expression into the previous gradient error bound yields Eq. \eqref{eq47}.

%








\end{document}